\documentclass[conference,twocolumn,10pt]{IEEEtran}
\IEEEoverridecommandlockouts
\pdfoutput=1
\usepackage{cite}
\usepackage{amsmath,amssymb,amsfonts,amsthm}
\usepackage{algorithmic}
\usepackage{graphicx}
\usepackage{textcomp}
\usepackage{xcolor}
\usepackage{hyperref}
\usepackage{prooftree}
\usepackage{fancybox}
\usepackage{txfonts}
\usepackage{proofzilla}
\usepackage{verbatim}

\usepackage{mathtools}

\usepackage{thmtools}
\usepackage{thm-restate}
\usepackage{cleveref}

\def\BibTeX{{\rm B\kern-.05em{\sc i\kern-.025em b}\kern-.08em
    T\kern-.1667em\lower.7ex\hbox{E}\kern-.125emX}}
\usepackage{xspace}
\usepackage{soul}
\usepackage{xspace}

\usepackage{empheq}
\usepackage[theorems,skins]{tcolorbox}

\usepackage{flushend}

\makeatletter
\newcommand{\xdashrightarrow}[2][]{\ext@arrow 0359\rightarrowfill@@{#1}{#2}}
\newcommand{\xdashleftarrow}[2][]{\ext@arrow 3095\leftarrowfill@@{#1}{#2}}
\newcommand{\xdashleftrightarrow}[2][]{\ext@arrow 3359\leftrightarrowfill@@{#1}{#2}}
\def\rightarrowfill@@{\arrowfill@@\relax\relbar\rightarrow}
\def\leftarrowfill@@{\arrowfill@@\leftarrow\relbar\relax}
\def\leftrightarrowfill@@{\arrowfill@@\leftarrow\relbar\rightarrow}
\def\arrowfill@@#1#2#3#4{%
  $\m@th\thickmuskip0mu\medmuskip\thickmuskip\thinmuskip\thickmuskip
   \relax#4#1
   \xleaders\hbox{$#4#2$}\hfill
   #3$%
}
\makeatother

\newif\ifcomments
\commentstrue 

\newcommand{\del}[1]{\ifcomments\hl{ }\st{#1}\hl{ }\fi}
\newcommand{\new}[1]{\sethlcolor{yellow}\hl{#1}}
\newcommand{\maxi}[1]{\ifcomments\sethlcolor{green}\hl{\textsc{Maxi}: #1}\fi}
\newcommand{\fred}[1]{\ifcomments\sethlcolor{cyan}\hl{\textsc{Fred}: #1}\fi}
\newcommand{\ross}[1]{\ifcomments\sethlcolor{pink}\hl{\textsc{Ross}: #1}\fi}

\newtheorem{definition}{Definition}
\newtheorem{theorem}{Theorem}
\newtheorem{corollary}{Corollary}
\newtheorem{lemma}{Lemma}

\newcommand{\flow}{\leadsto}
\newcommand{\defs}{\triangleq}

\newcommand{\lts}[1]{\xsta{#1}}
\newcommand{\xsta}[1]{ 
\setbox0=\hbox{\,${\scriptstyle#1}$\,}
\ifdim\wd0<12pt\wd0=12pt\fi
\mathrel{\raise2.2pt\hbox{$\mathop{\rule{\wd0}{0.6pt}}\limits^{#1}$}\mkern-7mu\blacktriangleright}
}
\newcommand{\by}[1]{{\footnotesize [by #1]}}

\newcommand{\Lts}[1]{ 
\setbox0=\hbox{\,${\scriptstyle#1}$\,}
\ifdim\wd0<12pt\wd0=12pt\fi
\mathrel{\raise1.6pt\hbox{$\mathop{\mathrlap{\rule{\wd0}{0.6pt}}{\rule[1.4pt]{\wd0}{0.6pt}}}\limits^{#1}$}\mkern-7mu\blacktriangleright}
}

\newcommand{\Read}[2]{\mathrm{R}(#1,#2)}
\newcommand{\Write}[2]{\mathrm{W}(#1,#2)}

\newcommand{\ds}[1]{\mathrm{ds}(#1)}
\newcommand{\coi}[1]{\mathrm{coi}(#1)}
\newcommand{\CoIC}{\mathcal{C}}

\newcommand{\setlog}{$\{log\}$\xspace}
\newcommand{\parameters}{\mathsf{parameters}}
\newcommand{\variables}{\mathsf{variables}}
\newcommand{\axiomsl}{\mathsf{axiom}}
\newcommand{\invariant}{\mathsf{invariant}}
\newcommand{\initial}{\mathsf{initial}}
\newcommand{\operation}{\mathsf{operation}}
\newcommand{\Theorem}{\mathsf{theorem}}
\newcommand{\Yo}{\mathit{Yo}}
\newcommand{\Xo}{\mathit{Xo}}
\newcommand{\N}{N}
\newcommand{\W}{W}
\newcommand{\RR}{\mathit{Sds}}
\newcommand{\siRR}{\mathit{minSub}}

\newcommand{\Pfun}{\mathsf{pfun}}
\newcommand{\Dom}{\mathsf{dom}}
\newcommand{\Ran}{\mathsf{ran}}
\newcommand{\ApplyTo}{\mathsf{applyTo}}

\newcommand{\Oplus}{\mathsf{oplus}}
\newcommand{\Diff}{\mathsf{diff}}
\newcommand{\Foreach}{\forall}
\newcommand{\simpSec}{\mathit{simpSec}}
\newcommand{\starProp}{\mathit{starProp}}
\newcommand{\strRead}{\mathit{spRead}}
\newcommand{\wRead}{\mathit{wkRead}}
\newcommand{\readWrite}{\mathit{readWrite}}
\newcommand{\rRead}{\mathit{rvkRead}}
\newcommand{\rqWrite}{\mathit{write}}
\newcommand{\init}{\mathit{init}}
\DeclareMathSymbol{\dres}{\mathbin}{letters}{"2F}
\def	\lover		{\mathbin{{\dres} \llap{+\!}\;}}
\renewcommand{\oplus}{\lover}

\begin{document}

\title{Brewer-Nash Scrutinised: Mechanised Checking of Policies featuring Write Revocation
}

\makeatletter
\newcommand{\linebreakand}{%
  \end{@IEEEauthorhalign}
  \hfill\mbox{}\par
  \mbox{}\hfill\begin{@IEEEauthorhalign}
}
\makeatother

\author{\IEEEauthorblockN{Alfredo Capozucca}
\IEEEauthorblockA{\textit{Department of Computer Science}\\
\textit{University of Luxembourg}\\
Esch-sur-Alzette, Luxembourg \\
alfredo.capozucca@uni.lu \\}
\and
\IEEEauthorblockN{Maximiliano Cristi{\'a}}
\IEEEauthorblockA{\textit{Universidad Nacional de Rosario}\\
\textit{CIFASIS} \\
Rosario, Argentina \\
cristia@cifasis-conicet.gov.ar \\}
\linebreakand 
\IEEEauthorblockN{Ross Horne}
\IEEEauthorblockA{\textit{Computer and Information Sciences} \\
\textit{University of Strathclyde}\\
Glasgow, United Kingdom \\
ross.horne@strath.ac.uk \\}
\and
\IEEEauthorblockN{Ricardo Katz}
\IEEEauthorblockA{\textit{CIFASIS CONICET} \\
Rosario, Argentina \\
katz@cifasis-conicet.gov.ar \\}}

\maketitle

\begin{abstract}
This paper revisits the Brewer-Nash security policy model inspired by ethical Chinese Wall policies. We draw attention to the fact that write access can be revoked in the Brewer-Nash model. The semantics of write access were underspecified originally, leading to multiple interpretations for which we provide a modern operational semantics. We go on to modernise the analysis of information flow in the Brewer-Nash model, by adopting a more precise definition adapted from Kessler. For our modernised reformulation, we provide full mechanised coverage for all theorems proposed by Brewer \& Nash. Most theorems are established automatically using the tool \setlog with the exception of a theorem regarding information flow, which combines a lemma in \setlog with a theorem mechanised in Coq. Having covered all theorems originally posed by Brewer-Nash, achieving modern precision and mechanisation, we propose this work as a step towards a methodology for automated checking of more complex security policy models.
\end{abstract}


\section{Introduction}

The Brewer-Nash security policy model, inspired by Chinese Wall policies, used to manage conflicts of interest particularly in the financial sector,
was originally communicated at S\&P'89~\cite{Brewer1989}.
Chinese Walls remain as much a feature of modern businesses, as when Brewer \& Nash motivated their work,
with the ongoing high-profile insider-trading case of Joe Lewis highlighting the importance of being able to provide evidence that adequate policies were adhered to.\footnote{See, e.g., FT on insider trading: \url{https://www.ft.com/insider-trading}}
In computer systems, Chinese Walls have been adopted since they allow freedom of choice initially, until too much information is requested.
There are established implementations of Chinese Wall policies for Unix~\cite{Foley1997},
and the progression of such policies to the Xen hypervisor, where multiple organisations may share the same hardware, was almost inevitable~\cite{Sailer2005}.

Brewer \& Nash deliberately designed their security policy model
such that features are reminiscent of the Bell-LaPadula security policy model~\cite{Bell1973}, that informed most lattice-based security policies.
These policy models are schemes for
security policies that maintain confidentiality (and sometimes integrity) of information by permitting or denying certain flows through a system.
While Bell-LaPadula, which originated in policies typical of the military, would permit flows from low to high classification of objects,
Brewer \& Nash proposed a flat structure, more typical of companies that do not have a common administrative authority.

A distinctive feature of Chinese Wall policies, sometimes referred to as ethical policies,
is that they provide some mechanism for explicitly indicating conflicts of interest at some granularity such as a dataset containing information about a specific company or legal entity.
A conflict of interest (CoI) can become problematic, for example, when there are nondisclosure agreements in place and
a single consultant can read confidential information about competitors for which consultancy services are offered.
CoIs remain a feature of more recent models of ethical policies,
such as quantales of information~\cite{Hunt2021}.
The year Brewer \& Nash communicated their model,
Lin presented a compelling argument that conflicts of interest should be represented by a general relation between datasets,
since, a conflict-of-interest relation need not be transitive (for example, if two readers have a conflict of interest with an author then the two readers are not necessarily in conflict with each other)~\cite{Lin1989}.
That idea, due to Lin, is now accepted in the mainstream, e.g., the aforementioned Unix implementation features a matrix capturing a CoI relation which need not be transitive.

Brewer \& Nash were likely aware that CoI is naturally more general in the sense of Lin, since they state explicitly that
``since we wish to compare it with the
Bell-LaPadula (BLP) model we will adopt the latter's
concepts of subjects, objects and security labels.''
Thus some of the restrictions in their model were 
more so to facilitate an easy comparison with concepts due in the Bell-LaPadula model.
There are other dimensions in which the Brewer-Nash model has 
evolved over time, becoming increasingly complex---e.g., permitting more flexible policies, or distinguishing between permission to access and instances of access, etc.~\cite{Kessler1992,Sharifi2013}.
The original model, due to Brewer \& Nash, however still survives in 
 information security textbooks~\cite{Anderson2020}.

In this work, we return our attention to 1989 and the original model of Brewer \& Nash.
While the model appears to be simple, there are some features that are not easy to grasp,
or are easily missed, since they
are handled in a rather implicit manner.
\begin{enumerate}
\item
Firstly, the notion of information flow that they rely on is not, in our view, as well defined as in later papers on security policy models.
\item
Secondly, the model has a rather novel feature for a security policy model: write access can be revoked, but in a rather implicit manner.
\end{enumerate}
Point (1) above is an indicator that, while an appendix with proofs was provided,
the rigour 
of the proofs conducted was perhaps not up to today's standards.
This argument applies whether or not one agrees with our view above on information flow,
thanks to the advances in automated tools and proof assistants.
Since the automated tool for proving decidable theorems in set theory, called \setlog (pronounced set log),
has been used to fully and quickly automate the checking of the Bell-LaPadula policy model~\cite{Cristia2021},
it is natural to ask whether that automation can be lifted to other security policy models in general,
and Brewer-Nash in particular in this work.
What we will see is that Brewer-Nash is more complex to verify, 
and instead we go for a hybrid approach where, 
for the theorem concerning information flow, \setlog proves an invariant,
while Coq is used to mechanise a proof 
confirming
 that the invariant is sufficient to establish the intended property.
Thus, in summary, our contributions in this direction are:
\begin{itemize}
\item Tightening of the definitions given by Brewer \& Nash that we felt necessary to establish their original theorems.
\item A fully mechanised proof of the tightened theorems using \setlog and Coq, with maximum work pushed to the automated \setlog component.
\end{itemize}
The work can also be seen as a seed for a methodology that may be more generally applied to security policy models that maintain confidentiality and integrity with respect to information flows. Given that related policy models are deployed in real systems and their failure can have serious consequences, as discussed at the top of the introduction, it is important to certify the correctness of security policy models.

Returning to point (2) above,
we clarify in a more explicit manner how write access is handled.
In the motivating section, next, we elaborate on
 an example where we ask what happens when 
 you can write to a dataset, and then request (read) access to another dataset.
The example helps us see complications associated with this scenario,
which we believe is the key novelty of the Brewer-Nash model
compared to some more recent ethical policies.
It is also a novelty with respect to key implementations,
for instance the aforementioned Unix implementation 
of a Chinese Wall policy bypasses this feature of Brewer-Nash 
and instead proposes its own mechanism for write access.
Observations such as this, lead us to believe that the Brewer-Nash model is not as simple as it first seems,
and hence may be open to misuse without stronger certification.

\paragraph*{Summary}
Section~\ref{sec:motivation} illustrates the key novelties and ambiguities in the operational semantics of the Brewer-Nash policy model,
firstly in a simpler form ignoring sanitized data.
Section~\ref{sec:bridge} provides a complete operational semantics for Brewer-Nash policies, including sanitized data, and introduces the target properties expected of the Brewer-Nash policy model.
Section~\ref{sec:model} explains how some properties are mechanised using the tool \setlog by expressing appropriate invariants.
Section~\ref{sec:if} defines an appropriate notion of information flow and explains how information flow is mechanised by combining \setlog and Coq.
Section~\ref{sec:min} completes the mechanisation of all theorems by showing how the remaining theorem can be established in \setlog via an explicit construction of an injective function.
Section~\ref{sec:future} highlights how the methodology may be adapted to other policy models in the future.

\section{Motivating scenarios: interpretations of access}\label{sec:motivation}

We explain here a simple scenario in which write revocation occurs, while refreshing our knowledge about the Brewer-Nash policy model.
In doing so, we examine the constraints on state transitions determining whether access is permitted to resources:
namely, the \textit{simple security rule} and the \textit{*-property}.
The scenario illustrates complications in the original Brewer-Nash model
which does not distinguish between read and write access in the state---there is only access.
When one has access, one may read henceforth, yet each write access is conditional on the state.
This statement can already be interpreted in multiple ways and hence we require more precision to resolve how read and write access are interpreted.
We also provide a more explicit formulation of the Brewer-Nash model that separately handles read and write access in the state,
which we argue is more amenable to implementation.
We omit sanitized data from this initial discussion to keep to the point.

\subsection{Diagrammatically}

Consider the simple state transition illustrated in Fig.~\ref{fig:eg:revoke}.
This small example already demonstrates a surprising feature of the Brewer-Nash model, specifically that write permissions can be revoked.
To follow the illustration observe that
there are two conflict-of-interest classes (CoIC) $\czero$ and $\cone$
which set up boundaries between datasets (the rounded regions) indicated by the solid (red) lines.
Intuitively, no subject should be able to hold
data originating from objects in datasets at either side of the red line.
Each dataset has one object in this example: $\ozero$, $\oone$, and $\otwo$.
In a prior state (not shown in the figure), the subject $s_1$ can access no object, and has free will 
to request access from any dataset, which is permitted by Brewer-Nash.
In this case, the subject has chosen to access $\oone$,  
resulting in the state to the left of Fig.~\ref{fig:eg:revoke},
where the two-headed 
arrow 
($\color{modcolor} \xleftrightarrow{\qquad}$) 
indicates that $s_1$ has read/write access to $o_2$.
\begin{figure}[h!]
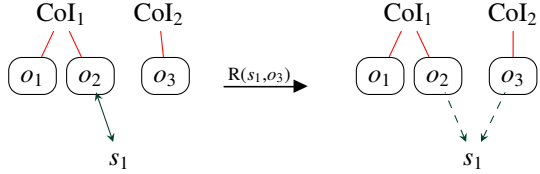

\[
\!\!\!\!\!\!
  \begin{array}{c}
  \begin{array}{c@{}c}    
    \vczero0 & \vcone1
    \\[1em]
  \begin{array}{c@{~}c}
  \vmod0{ ~\vozero0~  } &
  \vmod1{  ~\voone1~ }
  \end{array}
  &
  \begin{array}{c@{}c}
  \vmod2{ ~\votwo2~  }
  \end{array}
  \end{array}
  \\[3em]
  \qquad\vszero0 
  \end{array}
  \Wedges{szero0/oone1}
  \Redges{czero0/mod0,czero0/mod1,cone1/mod2}
\lts{\Read{\szero}{\otwo}}
  \begin{array}{c}
  \begin{array}{c@{}c}    
    \vczero0 & \vcone1
    \\[1em]
  \begin{array}{c@{~}c}
  \vmod0{ ~\vozero0~  } &
  \vmod1{  ~\voone1~ }
  \end{array}
  &
  \begin{array}{c}
  \vmod2{ ~\votwo2~  }
  \end{array}
  \end{array}
  \\[3em]
  \qquad\vszero0 
  \end{array}
  \Nedges{szero0/oone1}
  \Nedges{szero0/otwo2}
  \Redges{czero0/mod0,czero0/mod1,cone1/mod2}
\]
\caption{
A transition where write is revoked:
subject $\szero$ requests access to $\otwo$
(permitted by simple security),
which results in write access to $\oone$
being revoked (due to the *-property).
}\label{fig:eg:revoke}
\end{figure}
When in the state to the left of Fig.~\ref{fig:eg:revoke},
it is impossible for $\szero$ to access $\ozero$. 
This is the distinguishing feature of Chinese Wall policies in general---initially subjects have free choice to access objects,
but, as accesses are granted, permissions may be restricted.
In contrast, subject $\szero$ can request access to $\otwo$, which
lies in a separate CoIC.
Perhaps here, $\ozero$ and $\oone$ are each confidential data items in datasets of competitors, but $\otwo$ is data about a business operating in a separate market (c.f., banks v.s.\ oil companies in the original work of Brewer \& Nash).

Since $\szero$ is permitted to access $\otwo$,
 the system can perform the state transition labelled with the read request in Fig.~\ref{fig:eg:revoke}.
 However, notice that after that transition, the solid two-headed arrow ($\color{modcolor} \xleftrightarrow{\qquad}$) becomes a one way dotted arrow
 ($\color{modcolor} \xdashrightarrow{\qquad}$).
 These dotted arrows point from the object to the subject, indicating that, in the new state,
 (1) the subject has at some point been permitted access, (2) only read access is enabled in that state.
 Thus the subject can obtain information from the object, but not the other way round, as suggested by the direction of these arrows.
 The arrows therefore indicate that the write access of $\oone$ to $\szero$ has been revoked as soon as $\otwo$ is accessed, and furthermore $\szero$ is never permitted write access to $\otwo$.
 This write revocation is specific to a line of work faithful to Brewer-Nash~\cite{Lin1989,Kessler1992,Sharifi2013},
 but not all ethical policies.
 
\subsection{Operationally}
We explain more formally the machinery 
at play here,
to understand in what sense revocation of write access
 is handled implicitly by Brewer \& Nash,
  before we go on to express how write revocation may be made more explicit.

The Brewer-Nash policy model is a scheme for policies, where a specific Brewer-Nash policy is defined by security labels that are assigned to each object for the lifetime of the policy.
More precisely, a Brewer-Nash policy assumes that each object $o$
has fixed labels that assign the object a single CoIC, denoted here by $\coi{o}$,
and a dataset within that CoIC, denoted $\ds{o}$.
States are simply a (finite) relation between subjects and objects (the access matrix), denoted $N$.
We can make a state transition updating $N$ to 
indicate that a subject $s$ can access an object $o$,
only if $o'$ is in the same dataset as $o$ or is in a different CoIC from $o$.
This condition on access is called the \textit{simple security} rule.
\begin{definition}[simple security]\label{def:ss}
A subject $s$, object $o$ and matrix $N$ satisfy the simple security rule
whenever:
\[
\forall o' \colon (s, o') \in N \implies \left( \ds{o'} = \ds{o} \vee \coi{o} \neq \coi{o'} \right)
\]
\end{definition}
Thus read access, conditional on the simple security rule, can be
expressed, using the conventions of labelled transition systems,
as follows.
\begin{equation}\tag{iR}\label{eq:ss}
\begin{prooftree}
\forall o' \colon (s, o') \in N \implies \left( \ds{o'} = \ds{o} \vee \coi{o} \neq \coi{o'} \right)
\justifies
N \lts{\Read{s}{o}} N \cup \left\{ (s, o) \right\}
\end{prooftree}
\end{equation}
Stated otherwise, the above rule expresses that,
$s$ can access $o$ if $s$ has not accessed any other object that is in another dataset in the same CoIC as $o$.
It is not stated clearly or expressed formally in the original paper~\cite{Brewer1989}, but we can assume that 
``access'' here refers to
read access specifically.
The need for that clarification becomes important, given that the next property refers to
read and write access.


The *-property is described informally by Brewer \& Nash in terms of the capability to ``read'' and ``write'', rather than the neutral ``access'' of the simple security rule.
More precisely, it states that write access is permitted if
(1) the simple security rule holds,
and (2) no subject ``can read'' an object in a dataset different from the one requested.
The details of the operational rule for write access is open to interpretation.
We see arguments for and against an interpretation
where writing (disseminating and appending data) is possible before requesting read access using Eq.~\eqref{eq:ss}.

One interpretation, coming from
the model provided by Brewer-Nash in their appendix  
suggests that any access entails read access,
 since the formal way ``$s$ can read $o$'' is modelled
is by checking 
 whether $(s,o) \in N'$, where $N'$ is the state after the transition (see Axiom~6 in the original appendix of Brewer \& Nash~\cite{Brewer1989}).
%
The *-property implies the simple security rule
(hence checking the simple security rule is redundant as noted first by Lin~\cite{Lin1989}).
This leads us to the following labelled transition for write access.
\begin{equation}\tag{iRW}\label{eq:a}
\begin{prooftree}
\forall o' \colon (s, o') \in N \implies \ds{o'} = \ds{o}
\justifies
N \lts{\Write{s}{o}} N \cup \left\{ (s, o) \right\}
\end{prooftree}
\end{equation}
For example, assuming the policy given by the labels in Fig.~\ref{fig:eg:revoke},
transitions
$\left\{(\szero,\oone)\right\} \lts{\Write{\szero}{\oone}} \left\{(\szero,\oone)\right\}$
and 
$\left\{(\szero,\oone)\right\} \lts{\Read{\szero}{\oone}} \left\{(\szero,\oone)\right\}$
can be applied indefinitely at first,
until the transition
$\left\{(\szero,\oone)\right\} \lts{\Read{\szero}{\otwo}} \left\{(\szero,\oone),(\szero,\otwo)\right\}$
occurs.
After that point, transitions
$\left\{(\szero,\oone),(\szero,\otwo)\right\} \lts{\Read{\szero}{\oone}} \left\{(\szero,\oone),(\szero,\otwo)\right\}$
and
$\left\{(\szero,\oone),(\szero,\otwo)\right\} \lts{\Read{\szero}{\otwo}} \left\{(\szero,\oone),(\szero,\otwo)\right\}$
are enabled indefinitely.
Yet, no write transition involving $\szero$ is enabled.
Thus the write access of $\szero$ to $\oone$ is implicitly revoked.

Other authors have produced alternative interpretations for how write access is defined.
For example, in an influential lattice-based formulation of a Chinese Wall policy model~\cite{Sandhu93},
it is clear that Sandhu permits write access anywhere if read access is granted nowhere.
In short, according to Sandhu,
subjects are also labelled,
and those subjects with no read access are labelled with the bottom element in a lattice of security labels,
which is below all dataset labels that are assigned to objects (which, as normal, confine confidential information in objects to their dataset).
Since the *-property is generalised by Sandhu such that write is permitted upwards in a particular lattice,
clearly subjects that have not yet read anything can write anywhere.


We, the authors, even are split on how to interpret the definitions of Brewer \& Nash as operational rules (see their Def.~1 and Axiom~6),
and given there is a split in the literature, we explore both interpretations.
If we argue that granting write access does not automatically grant read access,
we can employ the following rule.\footnote{The rule can be formulated without the lattice machinery of Sandhu, who anyway does not formalise state transitions. Sandhu suggests only informally that privileges of a subject may float up the lattice to give the desired dynamics of the subject. The objective of Sandhu is to explain that Chinese Wall policy models can be cast in the same light as the Bell-Lapadula policy models when it comes to the simple security property and *-property; that is, 
subjects can read below and write above in a suitable lattice structure, depending on labels assigned to subjects as well as to objects.}
\begin{equation}\tag{iW}\label{eq:b}
\begin{prooftree}
\forall o' \colon (s, o') \in N \implies \ds{o'} = \ds{o}
\justifies
N \lts{\Write{s}{o}} N 
\end{prooftree}
\end{equation}
This models a more permissive policy allowing write access without granting read access, as indicated by not updating $N$.
 Thus, one can write freely to all datasets,
until one dataset is read from; at which point, write is (implicitly) revoked to all other datasets
(by the *-property).
Transitions illustrating this sequence of operations are presented in Fig.~\ref{fig:eg:write-only},
where the head of the arrow depicting write-only access points in the opposite direction from read-only access seen previously.
\begin{figure}[h!]
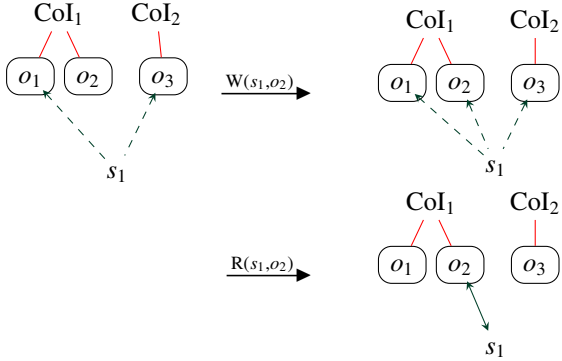

\[
\!\!\!\!\!\!
\begin{array}{rl}
  \begin{array}{c}
  \begin{array}{c@{}c}    
    \vczero0 & \vcone1
    \\[1em]
  \begin{array}{c@{~}c}
  \vmod0{ ~\vozero0~  } &
  \vmod1{  ~\voone1~ }
  \end{array}
  &
  \begin{array}{c@{}c}
  \vmod2{ ~\votwo2~  }
  \end{array}
  \end{array}
  \\[3.5em]
  \qquad\vszero0 
  \end{array}
  \Nedges{ozero0/szero0,otwo2/szero0}
  \Redges{czero0/mod0,czero0/mod1,cone1/mod2}
\lts{\Write{\szero}{\oone}} &
  \begin{array}{c}
  \begin{array}{c@{}c}    
    \vczero0 & \vcone1
    \\[1em]
  \begin{array}{c@{~}c}
  \vmod0{ ~\vozero0~  } &
  \vmod1{  ~\voone1~ }
  \end{array}
  &
  \begin{array}{c}
  \vmod2{ ~\votwo2~  }
  \end{array}
  \end{array}
  \\[3em]
  \qquad\vszero0 
  \end{array}
  \Nedges{ozero0/szero0,oone1/szero0,otwo2/szero0}
  \Redges{czero0/mod0,czero0/mod1,cone1/mod2}
\\ \lts{\Read{\szero}{\oone}} &
    \begin{array}{c}
  \begin{array}{c@{}c}    
    \vczero0 & \vcone1
    \\[1em]
  \begin{array}{c@{~}c}
  \vmod0{ ~\vozero0~  } &
  \vmod1{  ~\voone1~ }
  \end{array}
  &
  \begin{array}{c}
  \vmod2{ ~\votwo2~  }
  \end{array}
  \end{array}
  \\[3em]
  \qquad\vszero0 
  \end{array}
  \Wedges{szero0/oone1}
  \Redges{czero0/mod0,czero0/mod1,cone1/mod2}
\end{array}
\]
\caption{
Transitions in a permissive interpretation of Brewer-Nash policies
where write-only access can be requested anywhere initially.
Notice that write access is revoked when read access is requested.
}\label{fig:eg:write-only}
\end{figure}

Eq.~\eqref{eq:b} allows
each subject to have a phase where they use their own knowledge (to generate reports based on public information, for example)
and push it to any dataset.
This rule appears to be permissible from the perspective of information confidentiality, in the sense
that we expect that secrets held within objects (e.g., confidential information about clients that a trader must respect in the financial sector) will not be leaked to unintended objects by writing freely before reading anything confidential.
When a subject writes without reading, the subject 
can only write information known already before entering the ecosystem---only the subject's knowledge or special data processing skills are given away.
There is therefore no violation of a nondisclosure agreement with a company associated with a dataset,
since only the subject's own information is written, which is not subject to confidentiality constraints.
This keeps Brewer-Nash in line with the confidentiality aims of Bell-LaPadula avoiding inadvertent declassification,
but cannot prevent entirely people
going rogue and, for example, distributing 
 information outside the system (c.f. Teixeira's Pentagon leaks via screenshots posted on Discord).

This write-only (aka. append) phase precedes the phase of read/write editing within one dataset only.
And, finally, there is the read-only phase already discussed, where the information system
assists the subject in ensuring they do not violate a conflict of interests
in their ultimate decision making.
Confidentiality is therefore preserved (formalised in Sec.~\ref{sec:if}).

Both interpretations of write access discussed above have their merits.
Furthermore, we will find that it does not create problems going forward
to have both rules coexisting in one model.
Therefore, either rule, may be selected when implementing a system
enforcing a policy. Indeed, some objects may be governed by 
Eq.~\eqref{eq:a} and other by Eq.~\eqref{eq:b}
within the same system without compromising security.


\subsection{Explicitly}
From the above model we can see that the Brewer-Nash model is assuming that write is an atomic action,
in the sense that each time a subject would like
to write anything the *-property 
must be checked before the write access is granted,
and furthermore we must be sure that the write access completes before 
another operation is applied (or at least before a read request to another dataset is granted, in a system with more advanced awareness of concurrency).
This approach, by Brewer \& Nash to write operations
comes at a cost (in terms of concurrency control and logical checks), due to the need to check the *-property repeatedly while locking certain other operations,
rather than referring to an access control matrix.
That complexity is perhaps among the reasons why policies such as Unix Chinese Wall policies~\cite{Foley1997} do not implement write access using the *-property at all.

This observation leads us naturally to a more explicit approach that we introduce in this work,
which is
to include an additional write access matrix that makes any previously granted write access explicit until the *-property is violated.
Read and write access are formalised via the operational rules, explained next, that refer to a write access matrix $W$.

As explained when discussing Eq.~\eqref{eq:b},
an interpretation of Brewer-Nash
is that write access may be write-only,
and hence $N$ is not updated as in the following rule.
\begin{equation}\tag{xW}\label{eq:explicitW}
\begin{prooftree}
\forall o' \colon (s, o') \in N \implies \ds{o'} = \ds{o}
\justifies
N, W \lts{\Write{s}{o}} N, W \cup \left\{ (s,o) \right\}
\end{prooftree}
\end{equation}

Recall that,
 alternatively, 
a policy may insist that read access is granted whenever write access is granted (c.f., Eq.~\eqref{eq:a}),
which is a legitimate interpretation of the partial definitions of Brewer \& Nash.\footnote{
Derived from the explanations of Brewer \& Nash that: (1) $(s,o) \in N$ means ``$s$ has, or has had, access to object $o$.''
in a passage that can be interpreted as generically describing all types of access rather than read access specifically,
and; (2) there is no formal mention of read, except ``has read'' in the *-property and hence access is the only candidate;
(3) the state $N'$ after a write operation should contain the subject and object to which write access is granted.
}
In this case, observe below that both the 
read access and write matrix are updated,
where in what follows we define $W' = W \setminus \left\{ (s,o') \colon \ds{o'} \neq \ds{o} \right\}$.
\begin{equation}\tag{xRW}\label{eq:explicitRW}
\begin{prooftree}
\begin{array}{c}
\forall o' \colon (s, o') \in N \implies \ds{o'} = \ds{o}
\end{array}
\justifies
N, W \lts{\Write{s}{o}} N \cup \left\{ (s,o) \right\},
W' \cup \left\{ (s,o) \right\}
\end{prooftree}
\end{equation}
The updated write access matrix $W'$ 
explicitly revokes write
access, to any object in another dataset when the above rule is applied.
This caters for the possibility that write access may have been granted to another dataset,
which is clearly possible if Eq.~\eqref{eq:explicitW} is 
allowed to coexist with Eq.~\eqref{eq:explicitRW} above.\footnote{This will also be possible even if Eq.~\eqref{eq:explicitW} were forbidden, once we introduce sanitized data in the next section.
N.B.~``x'' abbreviates ``explicit''.
}


Since, once granted, read access is recorded in $N$ and write access is recorded in $W$,
we need not 
check the *-property each time a read or write occurs,
as in the original Brewer-Nash model.
Instead, we simply consult the access matrices $N$ or $W$ respectively and permit the operation if there is an appropriate entry.
This observation leads us to the following cheap rule for access, while other rules need only be appealed to if the rules below fail to grant access.
\begin{equation}\tag{access matrix}\label{eq:matrix}
\begin{prooftree}
(s,o) \in N
\justifies
N, W \lts{\Read{s}{o}} N, W
\end{prooftree}
\qquad
\begin{prooftree}
(s,o) \in W
\justifies
N, W \lts{\Write{s}{o}} N, W
\end{prooftree}
\end{equation}

The interesting question is
 what happens when read access is requested in another dataset
in a different CoIC from where write has been granted,
as per Fig.~\ref{fig:eg:revoke}.
Neither of the cheap access matrix lookup rules apply.
In this more explicit model,
in order to avoid a violation of the *-property,
it is important also to check that the new
read does result in the *-property being violated for some write that has already been granted.
If it does then either we:
\begin{itemize}
\item deny the read operation, or
\item explicitly revoke all offending write accesses in $W$. 
\end{itemize}
In terms of user experience, indeed it seems appropriate
to ask the user (or run some conflict resolution algorithm),
since it may be that the subject welcomes the warning and decides that they prefer not to read
the object in the new dataset, and instead retain read-write access to their current dataset.

\newtcolorbox{mymathbox}[1][]{colback=white, rounded corners, #1}

\begin{figure*}[h!]
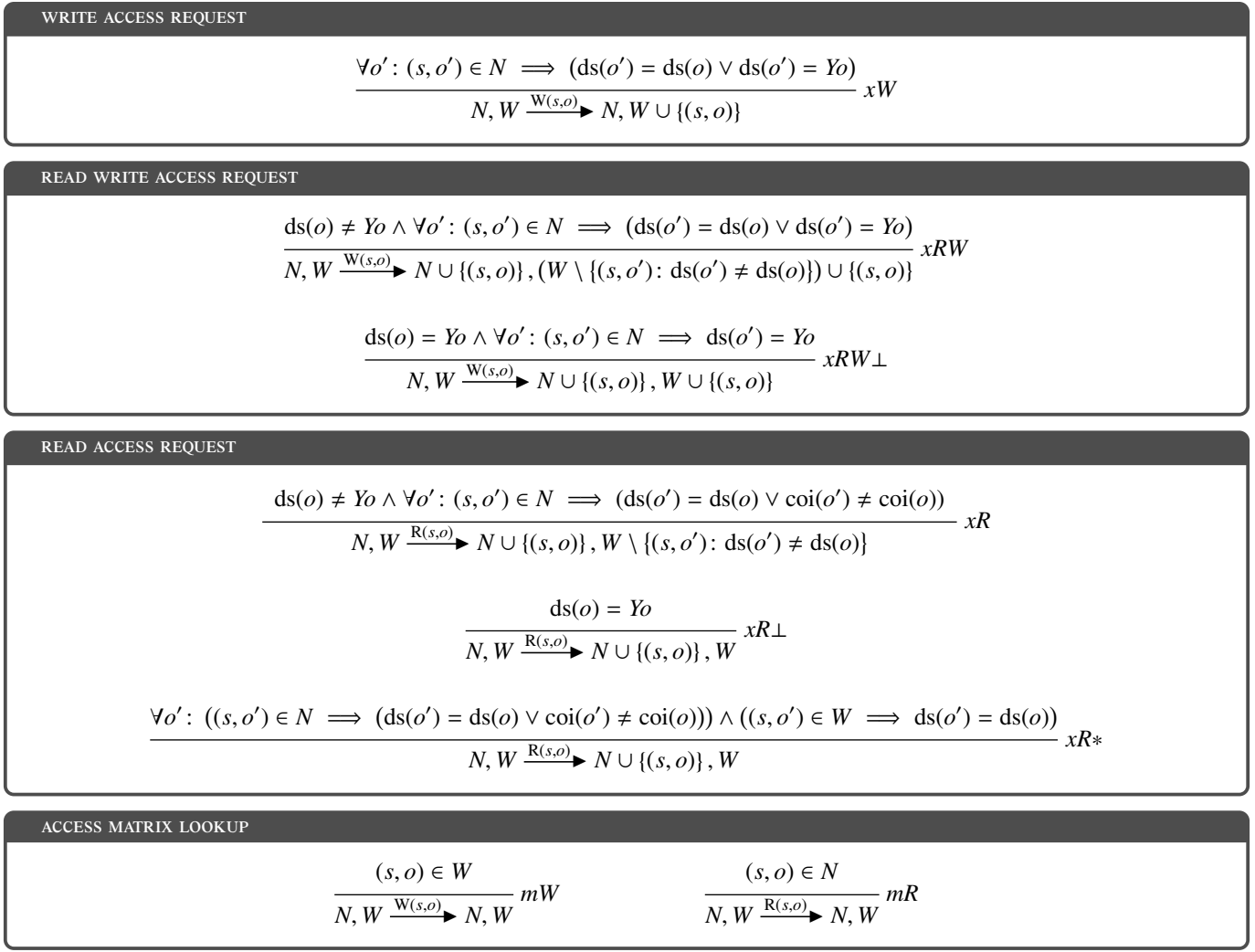

\begin{mymathbox}[ams gather*, title=\textsc{write access request}, colframe=white!30!black]
\begin{prooftree}
\forall o' \colon (s, o') \in N \implies \left( \ds{o'} = \ds{o} \vee \ds{o'} = \botDS \right)
\justifies
N, W \lts{\Write{s}{o}} N, W \cup \left\{ (s,o) \right\}
\using xW
\end{prooftree}
\end{mymathbox}

\begin{mymathbox}[ams gather*, title=\textsc{read write access request}, colframe=white!30!black]
\begin{prooftree}
\ds{o} \neq \botDS \wedge
\forall o' \colon (s, o') \in N \implies \left( \ds{o'} = \ds{o} \vee \ds{o'} = \botDS \right)
\justifies
N, W \lts{\Write{s}{o}} N 
\cup \left\{ (s,o) \right\}, 
\left(W \setminus \left\{ (s,o') \colon \ds{o'} \neq \ds{o} \right\} \right) \cup \left\{ (s,o) \right\}
\using xRW
\end{prooftree}
\\[15pt]
\begin{prooftree}
\ds{o} = \botDS \wedge
\forall o' \colon (s, o') \in N \implies \ds{o'} = \botDS
\justifies
N, W \lts{\Write{s}{o}} N 
\cup \left\{ (s,o) \right\}, 
W \cup \left\{ (s,o) \right\}
\using xRW\bot
\end{prooftree}
\end{mymathbox}

\begin{mymathbox}[ams gather*, title=\textsc{read access request}, colframe=white!30!black]
\begin{prooftree}
\begin{array}{c}
\ds{o} \neq \botDS
\wedge
\forall o' \colon (s, o') \in N \implies \left( \ds{o'} = \ds{o} \vee \coi{o'} \neq \coi{o} \right)
\end{array}
\justifies
N, W
\lts{\Read{s}{o}}
N \cup \left\{ (s,o) \right\}, W \setminus \left\{ (s,o') \colon \ds{o'} \neq \ds{o} \right\}
\using xR
\end{prooftree}
\\[15pt]
\begin{prooftree}
\ds{o} = \botDS
\justifies
N, W \lts{\Read{s}{o}} N \cup \left\{ (s,o) \right\}, W
\using xR\bot
\end{prooftree}
\\[15pt]
\begin{prooftree}
 \forall o' \colon
 \left( (s, o') \in N \implies \left( \ds{o'} = \ds{o} \vee \coi{o'} \neq \coi{o} \right) \right)
 \wedge
 \left(
 (s, o') \in W \implies \ds{o'} = \ds{o} \right)
\justifies
N, W \lts{\Read{s}{o}} N \cup \left\{ (s,o) \right\}, W
\using {xR*}
\end{prooftree}
\end{mymathbox}

\begin{mymathbox}[ams gather*, title=\textsc{access matrix lookup}, colframe=white!30!black]
\begin{prooftree}
(s,o) \in W
\justifies
N, W \lts{\Write{s}{o}} N, W
\using mW
\end{prooftree}
\qquad\qquad\qquad
\begin{prooftree}
(s,o) \in N
\justifies
N, W \lts{\Read{s}{o}} N, W
\using mR
\end{prooftree}
\end{mymathbox}
\caption{Explicit rules for Brewer-Nash policies, including sanitized data.}\label{fig:explicit-rules}
\end{figure*}

The two options above, correspond to
the following operational rule.
In the following, $W \setminus \left\{ (s,o') \colon \ds{o'} \neq \ds{o} \right\}$
explicitly revokes any offending write accesses.
\begin{equation}\tag{\small{xR}}\label{eq:explicit-strong-read}
\begin{prooftree}
\forall o' \colon (s, o') \in N \implies \left( \ds{o'} = \ds{o} \vee \coi{o'} \neq \coi{o} \right)
\justifies
N, W \lts{\Read{s}{o}} N \cup \left\{ (s,o) \right\},
 W \setminus \left\{ (s,o') \colon \ds{o'} \neq \ds{o} \right\}
\end{prooftree}
\end{equation}
Thus the rule above can be used to more explicitly realise the transition in Fig.~\ref{fig:eg:revoke}, as follows.
\[
\left\{ (\szero,\oone) \right\},
\left\{ (\szero,\oone) \right\}
\lts{\Read{\szero}{\otwo}}
\left\{ (\szero,\oone), (\szero,\otwo) \right\}, \emptyset
\] 
Similarly, the operations in Fig.~\ref{fig:eg:write-only}
consist of a write operation followed by a read that revokes write access to two datasets.
\[
\begin{array}{rl}
\emptyset,
\left\{ (\szero,\ozero), (\szero,\otwo) \right\}
\lts{\Write{\szero}{\oone}}
&
\emptyset,
\left\{ (\szero,\ozero), (\szero,\oone), (\szero,\otwo) \right\}
\\
\lts{\Read{\szero}{\oone}}
&
\left\{ (\szero,\oone) \right\},
\left\{ (\szero,\oone) \right\}
\end{array}
\]
An alternative is to allow a read access, without revoking write access,
under conditions 
ensuring that the *-property will be preserved for everything already recorded in $W$.
The condition is that all objects that the subject can write to according to $W$ must be in the same dataset as where read access is requested.
This condition concerning $W$ is of course in addition to the standard assumption that the simple security rule holds with respect to objects the subject can read from.
This restrictive read rule is expressed as follows.
\begin{equation}\tag{xR*}\label{eq:explicit-weak-read}
\begin{prooftree}
\begin{array}{c}
\forall o' \colon (s, o') \in N \implies (\ds{o'} = \ds{o} \vee \coi{o'} \neq \coi{o})
\\
\wedge
~ (s, o') \in W \implies \ds{o} = \ds{o'}
\end{array}
\justifies
N, W 
\lts{\Read{s}{o}} N \cup \left\{ (s,o) \right\}, W 
\end{prooftree}
\end{equation}
The ``*'' in the rule name above highlights the additional check required to preserve the *-property.
We will return to
the rules above in Eq's~\eqref{eq:explicit-strong-read} and~\eqref{eq:explicit-weak-read} 
in detail in subsequent sections,
since they are novel rules, and it is not immediately obvious that they do in fact preserve the *-property.
Thus our explicit rules benefit from the ensuing verification.

Notice that if Eq.~\eqref{eq:explicit-weak-read} applies, then Eq.~\eqref{eq:explicit-strong-read}
also applies and has the same effect.
However, if a policy features both rules then, whenever Eq.~\eqref{eq:explicit-weak-read} is not enabled
it is possible to trigger a suitable warning that 
explains to the subject that reading the object in question
 is going to result in write access being revoked somewhere else.
Thus, distinguishing these transitions
helps to demarcate an important transition in the life of a subject.
Eq.~\eqref{eq:explicit-weak-read} ensures a subject can continue to read and write within a dataset;
while if only Eq.~\eqref{eq:explicit-strong-read} 
applies then a subject induces a state transition that may prevent the subject from writing again.


\section{Full definitions and theorems to cover}\label{sec:bridge}

Here we collate the key definitions and theorems that we mechanise in this work.
In subsequent sections, we explain the theorems in more detail including how \setlog and Coq are used to mechanise them.
The theorems are the four theorems stated by Brewer-Nash in the order that they appear in that work to facilitate a close comparison.
These reformulated theorems make use of the modernised notation and definitions from the previous section.

\subsection{The full explicit model with sanitized data}

For complete coverage of Brewer-Nash we introduce the concept of sanitized data,
that didn't play a role in the previous section.
The explicit rules from the previous section are expanded and collated in Fig.~\ref{fig:explicit-rules}.

Sanitized data can refer to public information,
general market information, or, perhaps, data checked and approved to be distributed within the system without revealing confidential information of an entity regulated by the policy.
How data is sanitized is perpendicular to the Brewer-Nash model.
Indeed, there is a science of 
data sanitization, using ``association rules'' for example~\cite{Atallah1999},
that can determine how some confidential data may be sanitized for consumption beyond organisational boundaries.

In the Brewer-Nash security policy model,
a special dataset, denoted by $\botDS$, 
contains sanitized objects,
which is the only dataset in a conflict of interest class which we call
$\Sanitized$ (see Fig.~\ref{fig:eg:sanitized}).
The *-property
in the presence of sanitized data
is formulated as follows.
\begin{definition}[*-property]\label{def:sp}
A subject $s$, object $o$ and matrix $N$ satisfy the *-property
whenever:
\[
\forall o' : (s,o') \in N \implies \left( \ds{o} = \ds{o'} \vee \ds{o'} = \botDS \right)
\]
\end{definition}

Besides strengthening the *-property as shown above, 
we have handled the presence of santized data by splitting 
some rules from the previous section into multiple rules.
For example, we have the rules
 $\textit{xRW}$ and $\textit{xRW}\bot$
 for read-write access in Fig.~\ref{fig:explicit-rules}.
  Notice that clause
 $\ds{o} \neq \botDS$ ensures that  
 the transition $\textit{xRW}$ explicitly does not apply if
 the read-write request concerns an object in the santized dataset.
This is because the $\textit{xRW}$ rule is designed such that write access may be revoked to other datasets and revoking write access is not necessary when sanitized
 data is accessed. 
In contrast, the rule $\textit{xRW}\bot$ describes the effect of requesting read-write access to sanitized data,
where no revocation occurs.
The rules 
 $\textit{xR}$ and $\textit{xR}\bot$
 are separated for the same reason.
 We expand on this explanation next by providing an example where access to sanitized behaves differently.

\begin{figure}[h!]
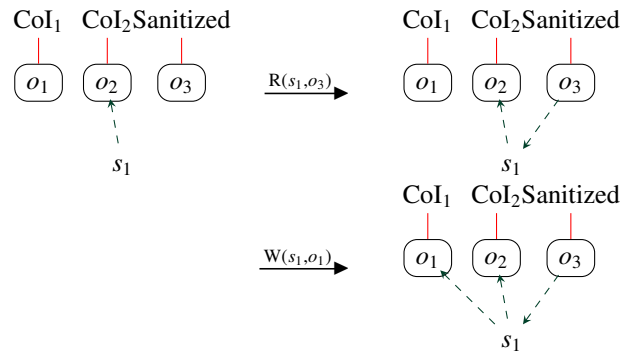

\[
\!\!\!\!\!\!
\begin{array}{rl}
  \begin{array}{c}
  \begin{array}{c@{~~}c@{}c}    
    \vczero0 & \vcone1
    & \vbot2
    \\[1em]
  \vmod0{ ~\vozero0~  } &
  \vmod1{  ~\voone1~ } &
  \vmod2{ ~\votwo2~  }
  \end{array}
  \\[3em]
  \vszero0 
  \end{array}
  \Nedges{oone1/szero0}
  \Redges{czero0/mod0,cone1/mod1,bot2/mod2}
\lts{\Read{\szero}{\otwo}} &
  \begin{array}{c}
  \begin{array}{c@{~~}c@{}c}    
    \vczero0 & \vcone1
    & \vbot2
    \\[1em]
  \vmod0{ ~\vozero0~  } &
  \vmod1{  ~\voone1~ } &
  \vmod2{ ~\votwo2~  }
  \end{array}
  \\[3em]
  \vszero0 
  \end{array}
  \Nedges{oone1/szero0,szero0/otwo2}
  \Redges{czero0/mod0,cone1/mod1,bot2/mod2}
\\ \lts{\Write{\szero}{\ozero}} &
  \begin{array}{c}
  \begin{array}{c@{~~}c@{}c}    
    \vczero0 & \vcone1
    & \vbot2
    \\[1em]
  \vmod0{ ~\vozero0~  } &
  \vmod1{  ~\voone1~ } &
  \vmod2{ ~\votwo2~  }
  \end{array}
  \\[3em]
  \vszero0 
  \end{array}
  \Nedges{ozero0/szero0,oone1/szero0,szero0/otwo2}
  \Redges{czero0/mod0,cone1/mod1,bot2/mod2}
\end{array}
\]
\caption{
Even after writing to a private dataset,
the user can read from sanitized data.
Furthermore, reading sanitized data does not prevent writing to further datasets.
}\label{fig:eg:sanitized}
\end{figure}

Consider the example in Fig.~\ref{fig:eg:sanitized} to understand how rules behave differently when sanitized data is involved.
The read transition in that figure uses rule $\textit{xR}\bot$,
which allows objects in the sanitized dataset to be read without revoking write access anywhere.
In contrast, rule \textit{xR} does not apply to this example, and if the condition $\ds{o} \neq \botDS$ were dropped from \textit{xR},
then \textit{xR} would revoke write access where it is not necessary to do so.
Fig.~\ref{fig:eg:sanitized}
also illustrates that write access can still be freely requested elsewhere
after reading.
This behaviour contrasts to read access to data that is not sanitized,
which always blocks write operations in other datasets from that point onward.



The notation $\bot$ in rule name signals the consistency of 
Brewer-Nash with the lattice-based interpretation of Sandhu~\cite{Sandhu93},
mentioned in the previous section.
Sandhu assigns for
sanitized objects, and also subjects who have not read from a dataset that is not sanitized,
the bottom security label in a lattice of labels.
Therefore, since in Sandhu's lattice-based model reading is permitted downwards and writing upwards,
subjects with the bottom security label can still request read access to sanitized data while writing anywhere.


Consider now the example in Fig.~\ref{fig:eg:sanitized-rw}.
This shows an exceptional revocation behaviour associated with sanitized data.
Initially, subject $\szero$ can read and write to a sanitized object.
The subject can, by calling rule $xRW$,
induce a state change where their write access to the sanitized object is revoked.
Observe that write access is enabled for object $\oone$.
This contrasts to Fig.~\ref{fig:eg:revoke} which did not involve
sanitized data and for which only read access was enabled after reading from two datasets.


\begin{figure}[h!]
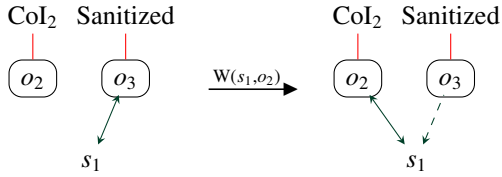

\[
\!\!\!\!\!\!
  \begin{array}{c}
  \begin{array}{c@{~~}c}    
     \vcone1
    & \vbot2
    \\[1em]
  \vmod1{  ~\voone1~ } &
  \vmod2{ ~\votwo2~  }
  \end{array}
  \\[3em]
  \vszero0 
  \end{array}
  \Wedges{otwo2/szero0}
  \Redges{czero0/mod0,cone1/mod1,bot2/mod2}
\lts{\Write{\szero}{\oone}} 
  \begin{array}{c}
  \begin{array}{c@{~~}c}    
   \vcone1
    & \vbot2
    \\[1em]
  \vmod1{  ~\voone1~ } &
  \vmod2{ ~\votwo2~  }
  \end{array}
  \\[3em]
  \vszero0 
  \end{array}
  \Nedges{szero0/otwo2}
  \Wedges{oone1/szero0}
  \Redges{czero0/mod0,cone1/mod1,bot2/mod2}
\]
\caption{
Only sanitized data has the property
that if read-write access has been granted,
then read-write access may still be granted in another dataset.
}\label{fig:eg:sanitized-rw}
\end{figure}

\subsection{\label{bnthm}The four Brewer-Nash Theorems}

We state the four theorems proposed by Brewer \& Nash in their original form and using more modern precision.
All of these theorems are stated with respect to the labelled transition system of the explicit model generated by the rules in Fig.~\ref{fig:explicit-rules}.
Thus, all theorems in this section
range over any policy consisting of
a fixed set of subjects, objects, CoICs,
mappings $\ds{\,.\,}$, and $\coi{\,.\,}$,
and satisfy axioms:
\begin{gather}
\forall o_1, o_2 \colon \ds{o_1} = \ds{o_2} \implies \coi{o_1} = \coi{o_2} 
\label{ax1}
\\
\forall o \colon \ds{o} = \botDS \Leftrightarrow \coi{o} =
\Sanitized 
\label{ax2}
\end{gather}
Axiom~\eqref{ax1}, above, ensures that all objects in a dataset are also in the same CoIC. 
Axiom~\eqref{ax2} formalises the requirement that there is a unique dataset and CoIC for sanitized data.
A consequence of Axiom~\eqref{ax2} 
is that, when sanitized data is considered,
whenever a subject, object and matrix 
satisfy the *-property then they satisfy the simple security rule. 
Both directions of the implication in Axiom~\eqref{ax2} are necessary to establish that fact.
This helps explain why the simple security rule is redundant in the rules in Fig.~\ref{fig:explicit-rules} involving write access.

For the first theorem, Brewer \& Nash state: ``Once a subject has accessed an object the only
other objects accessible by that subject lie within
the same company dataset or within a different
conflict-of-interest class.''

This essentially says that, w.r.t.~Fig.~\ref{fig:explicit-rules}, 
it is an invariant
that all subject-object pairs in $N$ 
satisfy the simple security rule.
A property $I$ is proven to be an invariant 
by establishing that 
if $I(N,W)$ holds and, for any action $\alpha$ and state $N', W'$, if
$N, W \lts{\alpha} N', W'$,
then 
$I(N', W')$ holds.
Furthermore, the initial state $\emptyset, \emptyset$
must satisfy $I$.
This ensures that the property holds in all states reachable from the initial state.
\begin{theorem}\label{thm1}
The following is an invariant for states $N, W$:
for all $(s,o) \in N$,
$s$, $o$ and $N$ satisfy the simple security rule.
\end{theorem}


The second theorem of Brewer \& Nash states: ``A subject can at most have access to one company
dataset in each conflict-of-interest class.''
This is captured formally by the consequent of Theorem~\ref{thm2} below.
\begin{theorem}\label{thm2}
Consider any state $N, W$ satisfying the property that,
for all $(s,o) \in N$,
$s$, $o$ and $N$ satisfy the simple security rule.
For any subject $s$ and objects $o_1$ and $o_2$,
if $(s, o_1) \in N$ and $(s, o_2) \in N$ 
and 
$\coi{o_1} = \coi{o_2}$,
then 
$\ds{o_1} = \ds{o_2}$.
\end{theorem}
The premise of the Theorem~\ref{thm2}, was absent in the formulation due to Brewer \& Nash.
The premise clarifies that the Theorem holds for any state satisfying the invariant established in Theorem~\ref{thm1}.
Hence the consequent of Theorem~\ref{thm2} is itself preserved by all transitions, and hence is itself an invariant.

The third theorem of Brewer \& Nash states: 
``If for some conflict-of-interest class $X$ there are
$X_v$ company datasets then the minimum number of
subjects which will allow every object to be accessed
by at least one subject is $X_v$.''
The intention of this theorem is to explain to a manager
how many subjects, e.g., consultants, are required to serve all datasets, e.g., companies,
without violation of CoIs.

In the above, ``accessed by'' is interpreted as read access as recorded by $N$.
This leads to the following formalisation of this property as an invariant, where 
$\lvert A \rvert$ denotes the cardinality of set $A$.
\newcommand{\X}{X}
\newcommand{\XV}{X_v}
\begin{theorem}\label{thm3}
Let $\mathcal{S}$ be some set of subjects
and $\mathcal{D}$ be some set of datasets
 fixed for the policy.
Also, for any CoIC $\X$,
let $\XV = \left\{ Y \in \mathcal{D} \colon \exists o. \ds{o} = Y \wedge \coi{o} = \X \right\}$.

The following is an invariant for states $N,W$.
For any CoIC $\X$, if
for all datasets $Y \in \XV$
there exists subject $s$ and object $o$ such that $(s, o) \in N$ and $\ds{o} = Y$,
then 
$\left| \XV \right| \leq \left| \mathcal{S} \right|$.
\end{theorem}

Theorem 4 of Brewer \& Nash is formulated as:
``The flow of unsanitized information is confined to
its own company dataset; sanitized information may
however flow freely throughout the system.''

The notion of flow is not defined by Brewer \& Nash, 
and we defer the definition that we will employ until Sec.~\ref{sec:if}.
However, we assume here that there is some well-defined notion of
flow of information from one object $o_1$ to another object $o_2$,
starting from a given state $N, W$, denoted $o_1 \leadsto o_2$ in the following theorem.
\begin{restatable}{theorem}{theorem4}\label{thm4}
The following is an invariant for states $N, W$:
For objects $o$ and $o'$,
if $o \leadsto o'$ starting in $N, W$, then $\ds{o} = \botDS$ or $\ds{o} = \ds{o'}$.
\end{restatable}
Brewer \& Nash, in their proof of Theorem~\ref{thm4}, we believe, just jump to their desired 
conclusion by defining a relation that satisfies a given property.
Theorem~\ref{thm4} turns out to be the trickiest theorem to define more precisely and prove, as we explain in detail in Sec.~\ref{sec:if}.


\section{Automated reasoning about invariants in \setlog
}\label{sec:model}





In this section we show how we use \setlog to formally specify and verify
that the simple security rule and *-property are invariants of 
our formal interpretations of Brewer-Nash policies.
This section covers the mechanisation of
Theorem~\ref{thm1} and Theorem~\ref{thm2}  
and lays essential groundwork towards the mechanisation of Theorem~\ref{thm4} (Sec.~\ref{sec:if}).

The tool \setlog  is a constraint logic programming (CLP) language and satisfiability solver implemented in Prolog where finite sets are first-class citizens \cite{setlog,Dovier00}. The tool implements decision procedures for several fragments of set theory and set relation algebra \cite{DBLP:journals/jar/CristiaR20,DBLP:journals/jar/CristiaR21a,DBLP:journals/tplp/CristiaR23,10.1145/3625230,DBLP:journals/corr/abs-2208-03518}. A few in-depth empirical studies provide evidence that \setlog is able to solve non-trivial problems, e.g. \cite{DBLP:journals/jar/CristiaR21b,DBLP:journals/jar/CristiaR21}. On top of its CLP language, \setlog provides a state machine specification language (SMSL) inspired in the B notation \cite{Abrial00}. A verification condition generator (VGC) can then be used to automatically generate verification conditions (VC) ensuring that the state machine verifies some properties \cite[Sect. 11]{Rossi00}. \setlog inherits many Prolog features. For instance, variables must begin with a capital letter; the main program building block are predicates expressed as Horn clauses of the form \fbox{$head(params) \text{ :- } body.$} where $body$ is a \setlog formula (note the dot at the end of $body$).

We describe here the \setlog formalisation of the Brewer-Nash policy model,
available in the companion replication package~\cite{rep-package}.
The set of objects of the system ($\mathit{Objects}$), the function mapping objects onto security classes ($L$) and the dataset containing sanitized information ($\Yo$) and its conflict of interest class ($\Xo$, aka.~$\Sanitized$ in the previous section) are introduced as parameters of the model.
\begin{align*}
\parameters([\mathit{Objects},L,\Yo,\Xo]).
\end{align*}
The state space of the system is given by two state variables: $\N$, denoting the current read accesses for each subject; and $\W$, denoting the current write accesses for each subject.
\begin{align*}
& \variables([\N,\W]).
\end{align*}

Axioms are used to state properties of parameters. For example, $L$ is a function whose domain is $\mathit{Objects}$.\footnote{Instead of using the exact \setlog ASCII notation, we rather use a more math-oriented one thus avoiding some syntactic nuisances.}
\begin{align*}
& \axiomsl(axiomL). \\
& axiomL(L,\mathit{Objects}) \text{ :- } \Pfun(L) \land \Dom(L,\mathit{Objects}).
\end{align*}
Above, $\Pfun$ is a \setlog constraint stating that its argument is a function whereas $\Dom$ states that $\mathit{Objects}$ is the domain of $L$. 
Since finite sets are the main data structure in \setlog,
 \setlog admits sets of ordered pairs, i.e., binary relations. 

Relations $\N$ and $\W$,
as we will shortly see,
are augmented with the labels associated with objects (the dataset and CoIC) because \setlog proofs become faster.
This is a difference between purely theoretical considerations such as those discussed in Sec.~\ref{sec:motivation} and~\ref{sec:bridge} and the representation of those concepts in an automated tool. Invariants are used to ensure that labels in $\N$ and $\W$ are subject to the conditions imposed on $L$.

State invariants are given as predicates that depend on parameters and state variables.
An invariant property 
appealing to 
the simple security rule is encoded as follows in \setlog.
\begin{align*}
& \invariant(\simpSec). \\
& \simpSec(\N) \text{ :- } \\
& \quad \Foreach (S_1,(O_1,(C_1,D_1))), (S_2,(O_2,(C_2,D_2))) \in \N : \\
& \qquad S_1 = S_2 \implies (C_1 \neq C_2 \lor D_1 = D_2).
\end{align*}
That is, $\N$ is a set of ordered pairs of the form $(S,(O,\ell))$ where $S$ is a subject, $O$ an object and $\ell$ a security label (which in turn is of the form $(C,D)$ for some CoIC $C$ and dataset $D$). Then, if $(S,(O,\ell)) \in \N$ it means that subject $S$ is accessing object $O$ in read mode and the security label of $O$ is $\ell$. In this way, $\simpSec$ states that, if a subject is accessing two or more objects in read mode, their CoIC are different or their datasets are the same.
It is easy to check that $\simpSec$ is a faithful formalisation of the invariant in Theorem~\ref{thm1}.

Note that in $\simpSec$ the quantification is a restricted quantification made with ordered pairs instead of variables. A restricted quantification is a formula of the form $\forall x \in A: \phi$ equivalent to $\forall x(x \in A \implies \phi)$. The presence of ordered pairs as quantified expressions is a distinctive feature of \setlog which allows 
us to increase the decidable fragment of formulas featuring restricted quantifiers~\cite{DBLP:journals/corr/abs-2208-03518}.


An invariant preserving the 
*-property for all pairs in $W$ is defined as follows.
The preservation of this invariant will be used in Sec.~\ref{sec:if} as part of the proof of Theorem~\ref{thm4}.
\begin{align*}
& \invariant(\starProp). \\
& \starProp(\Yo,\N,\W) \text{ :- } \\
& \quad \Foreach (S_1,(O_1,(C_1,D_1))) \in N; (S_2,(O_2,(C_2,D_2))) \in \W : \\
& \qquad S_1 = S_2 \implies (D_1 = D_2) \lor D_1 = \Yo.
\end{align*}
That is, $\W$ has a similar structure to $\N$ although its interpretation is that subject $s$ is accessing object $o$ in write mode. In this way, $\starProp$ states that if a subject accesses some objects in read mode and others in write mode then they must belong to the same dataset or the subject is reading only sanitized information ($\Yo$). As with $\N$, the property where $(S,(O,\ell)) \in W$ implies $(O,\ell) \in L$, is stated as an invariant.

After giving all the invariants the initial state can be defined.
This states that no access is granted to any subject initially.
\begin{align*}
& \initial(\init). \\
& \init(\N,\W) \text{ :- } \N = \emptyset \land \W = \emptyset.
\end{align*}

Now state transitions, called operations, are specified. Operations are predicates depending on at least one state variable. If state variable $X$ is changed during the transition its new value is denoted by $X'$.
Operations are given by specifying their pre- and post-conditions as \setlog formulas.
The first operation we show corresponds to a model where read access is granted only if simple security and *-property are preserved.
This is called \emph{*-property read}, 
denoted
here $\strRead$, and corresponding to \textit{xR*} in Fig.~\ref{fig:explicit-rules} (it also incorporates \textit{xR}$\bot$).
\begin{align*}
& \operation(\strRead). \\
& \strRead(\Xo,\Yo,L,\N,\W,S,O,\N') \text{ :- } \\
& \quad (S,(O,(C,D))) \notin N \\
& \quad\land \ApplyTo(L,O,(C,D)) \\
& \quad\land \Foreach (S_1,(O_1,(C_1,D_1))) \in \N: \tag{$pre_{ss}$}\\
& \quad\qquad S_1 = S \implies (C_1 \neq C \lor D_1 = D) \\
& \quad\land (D = \Yo \tag{$pre_{sp}$} \\
& \quad\qquad\lor \Foreach (S_1,(O_1,(C_1,D_1))) \in \W: \\
& \quad\qquad\qquad S_1 = S \implies D_1 = D) \\
& \quad\land \N' = \{(S,(O,(C,D)))~/~\N\}. \tag{$post$}
\end{align*}
In the above, $\strRead$ takes $\Xo$, $\Yo$, the state variables, a subject ($S$) and an object ($O$), and returns $N'$, i.e. the new value of $N$. All but the last line are pre-conditions. The first precondition ensures that $S$ has not opened $O$ for reading. The second pre-condition states that the security label of $O$ is $(C,D)$ by using the \setlog constraint $\ApplyTo$. Pre-condition $pre_{ss}$ checks the simple security rule (Def.~\ref{def:ss}).
Pre-condition 
$pre_{sp}$ is necessary to preserve the *-property.
It ensures that $O$ is a sanitized object or that the dataset of the objects that $S$ has write access to coincides with the dataset of $O$.
If all these hold, then $(S,(O,(C,D)))$ is added to $\N$ by means of an extensional set constructor available in \setlog. In effect, since $\{X~/~A\}$ is interpreted as $\{X\} \cup A$, then $\N'$ is equal to $\N$ plus the ordered pair in question. Given that $\strRead$ grants read permission, $\W'$ is not included as an argument, thus $\W = \W'$.

The \setlog code includes two more variants of the read operation. In one of
them, called 
 $\wRead$, the $pre_{sp}$ pre-condition
is not present. 
Operation $\wRead$ can result in a conflict of interests,
as illustrated in 
Figure \ref{fig:eg:invalid-if}.
Notice that the
initial accesses of $\szero$ and $\sone$ respect the simple security invariant,
 since $\ozero$ and $\otwo$ are in different CoI classes.
After $\szero$ requests this excessively ``weak'' read access to $\otwo$
(since only the simple security pre-condition holds) $\szero$ can access privileged
information stored in $\ozero$ which should not be allowed as $\szero$ already gained
access to $\oone$. The access to privileged information is facilitated by $\sone$ as
it may transfer information from $\ozero$ and store it into $\otwo$.
%
Such shortcomings are identified by \setlog when reporting that $\wRead$ does not
preserve $\starProp$
when attempting its
mechanised verification.\footnote{Guidelines on how to reproduce the automated
verification are provided in the replication package~\cite{rep-package}.}


\begin{figure}[h!]
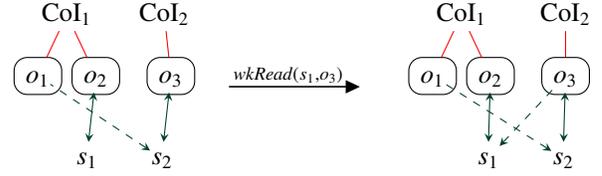

\[
\!\!\!\!\!\!
  \begin{array}{c}
  \begin{array}{c@{}c}    
    \vczero0 & \vcone1
    \\[1em]
  \begin{array}{c@{~}c}
  \vmod0{ ~\vozero0~  } &
  \vmod1{  ~\voone1~ }
  \end{array}
  &
  \begin{array}{c@{}c}
  \vmod2{ ~\votwo2~  }
  \end{array}
  \end{array}
  \\[3em]
  \qquad\vszero0
  \qquad\vsone1 
  \end{array}
  \Wedges{sone1/otwo2}
  \Nedges{sone1/ozero0}
  \Wedges{szero0/oone1}
  \Redges{czero0/mod0,czero0/mod1,cone1/mod2}
\lts{wkRead({\szero},{\otwo})}
  \begin{array}{c}
  \begin{array}{c@{}c}    
    \vczero0 & \vcone1
    \\[1em]
  \begin{array}{c@{~}c}
  \vmod0{ ~\vozero0~  } &
  \vmod1{  ~\voone1~ }
  \end{array}
  &
  \begin{array}{c}
  \vmod2{ ~\votwo2~  }
  \end{array}
  \end{array}
  \\[3em]
  \qquad\vszero0
  \qquad\vsone1 
  \end{array}
  \Wedges{sone1/otwo2}
  \Nedges{szero0/otwo2}
  \Nedges{sone1/ozero0}
  \Wedges{szero0/oone1}
  \Redges{czero0/mod0,czero0/mod1,cone1/mod2}
\]
\caption{
A conflict of interests resulting from executing the $\wRead$ operation.
}
\label{fig:eg:invalid-if}
\end{figure}

A variant of read that does preserve our invariants,
 called \emph{revoke read}, named $\rRead$ in the \setlog code and \textit{xR} in Fig.~\ref{fig:explicit-rules},
does not
contain $pre_{sp}$, as $\wRead$.
Instead, however, it updates $\W$ by revoking all the write
accesses of $S$ to objects whose dataset is different from the dataset of $O$. The update on
$\W$ is specified as follows.
\begin{align*}
& (D = \Yo \land W' = W \\
& \lor D \neq \Yo \\
& \quad \land \Diff(\W,\{(S_1,(O_1,(C_1,D_1))) \in \W \mid S_1 = S \land D_1 \neq D\},W'))
\end{align*}
Above, $\Diff(A,B,C)$ is a \setlog constraint interpreted as $C = A \setminus B$. $W'$ must be an argument of $\rRead$.
The condition $D \neq \Yo$ ensures that reading sanitized data does not result in write access being revoked. 

The specification of the operation \textit{xW} in Fig.~\ref{fig:explicit-rules},
called $\rqWrite$ in \setlog, is as follows.
\begin{align*}
& \operation(\rqWrite). \\
& \rqWrite(\Yo,L,\N,\W,S,O,\W') \text{ :- } \\
& \quad (S,(O,(C,D))) \notin W \\
& \quad\land \ApplyTo(L,O,(C,D)) \\
& \quad\land \Foreach (S_1,(O_1,(C_1,D_1))) \in \N: \tag{$pre_{sp}^w$} \\
& \quad\qquad S_1 = S \implies (D_1 = D \lor D_1 = \Yo) \\
& \quad\land \W' = \{(S,(O,(C,D)))~/~\W\}.
\end{align*}
In order to ensure that *-property is preserved after $\rqWrite$, $pre_{sp}^w$
checks that the objects already opened in read mode by $S$ belong to the dataset of $O$ or all of them contain only sanitized information.

There exists a second variant of the write operation, called \emph{read-write}, written $\readWrite$,
covering both \textit{xRW} and \textit{xRW}$\bot$ in Fig.~\ref{fig:explicit-rules}. While this operation ensures that *-property is
preserved similarly as $\rqWrite$, it also updates $\N$ and $\W$. The new
request---i.e. $(S,(O,(C,D)))$--- is added to $\N$ while $\W$ is updated either by
revoking all the write accesses of $S$ to objects whose dataset is different
from the dataset of $O$---as done in operation $\rRead$ shown earlier---or by adding the request when $S$ is accessing only sanitized information.

Once the operations have been given, the VCG is run thus generating a new file containing a number of VCs. Among the most important VCs are the so-called \emph{invariance lemmas}. An invariance lemma is a VC of the form $I \land T \implies I'$, where $I$ is an invariant, $T$ an operation and $I'$ is $I[\forall v \in st(I): v \mapsto v']$ with $st(I)$ the set of state variables of $I$. Informally, an invariance lemma states that if an invariant holds in some state and an operation is executed, the invariant holds in the next state. In other words, the invariant is \emph{preserved} by the operation. Given that \setlog is a satisfiability solver, invariance lemmas generated by the VCG take a negated form:
$\lnot(I \land T \implies I')$.
Hence, if \setlog determines that the above is unsatisfiable, the inner formula is a theorem. The VCG generates an invariance lemma for each invariant and operation. For example:
\begin{align*}
& \lnot(\starProp(\Yo,\N,\W) \land 
\rqWrite(\Yo,L,\N,\W,S,O,\W') \\
& \quad\implies \starProp(\Yo,\N,\W'))
\end{align*}
Note that in the consequent $\N$ appears rather than $\N'$, since $\N$ remains unchanged in $\rqWrite$ 
($N'$ is not one of its arguments).

Since an invariance lemma trivially holds if either the invariant or the
operation are unsatisfiable, \setlog also generates VCs ensuring that the
initial state satisfies every invariant and that all operations are
satisfiable.
The VCG generates also a predicate calling all the VCs. Then, when the user runs this predicate \setlog attempts to discharge all the VCs.

Besides the standard VCs concerning the verification of state machines, \setlog
users can define their own VCs in the form of clauses declared as
$\textsf{theorem}$. As with invariance lemmas, \setlog theorems have to be written in
negated form. Each such declaration is included by the VCG as a VC. User-defined
theorems have been used to prove, for instance, Theorems~\ref{thm1} and~\ref{thm2}. 
For Theorem~\ref{thm2} we first define a clause with its consequent
($t2$) and then we declare a theorem ($theorem2$) where $\simpSec$ is the
hypothesis required to prove $t2$. We can use $\simpSec$ as an hypothesis because
we have proved that it is a state invariant.
The fact that $\simpSec$ is enough to prove these theorems shows the importance
of finding the right invariants for a model and, more specifically, the
importance of $\simpSec$ and $\starProp$ in this context. This is further
stressed in Sec.~\ref{sec:if}.

\section{Information flow}\label{sec:if}

A goal of the Brewer-Nash policy model is to ensure that sensitive information
flows within its intended context.
This section explains the mechanisation of Theorem~\ref{thm4} 
that establishes confidentiality with respect to certain information flows resulting from read and write operations.
We formalise information flow and prove in Coq the confidentiality property expressed by Brewer \& Nash in Theorem~\ref{thm4}, using properties automatically discharged by \setlog.
Properties that \setlog discharges
include that (1) the *-property is an invariant,
(2) read access monotonically increases.
We settle for a definition of information flow based on 
Kessler's~\cite{Kessler1992}, the earliest definition of information flow in the context of Chinese Wall policies 
that is precise enough for our purposes.

\subsection{Defining information flows}

When defining information flow we make use of big-step labelled transitions,
that perform zero or more transitions 
before the given label occurs.
\begin{definition}[big-step transition]
$N_0, W_0 \Lts{\alpha_n} N_{n+1}, W_{n+1}$
whenever there exists $N_1, W_1, \ldots N_{n}, W_{n}$
and
for all $i\in [0\ldots n]$,
$N_i, W_i \lts{\alpha_i} N_{i+1}, W_{i+1}$ (according to Fig.~\ref{fig:explicit-rules}).
\end{definition}

We can now express formally a suitable notion of information flow,
inspired by Kessler~\cite{Kessler1992}.

\begin{definition}[information flow]\label{def:if}
Starting in state $N_1, W_1$,
information can flow from object $o_1$ to object $o_{n+1}$, written $o_1 \flow o_{n+1}$ whenever:
\begin{align*}
& \exists s_1,\dots,s_n; o_2,\dots,o_{n}; N_2,\dots,N_{2n+1}; W_2,\dots,W_{2n+1}: \\
& \quad \forall i : 1 \leq i \leq n \implies \\
& \quad\qquad
  N_{2i-1},W_{2i-1} \Lts{\Read{s_{i}}{o_{i}}} N_{2i},W_{2i} 
  \Lts{\Write{s_i}{o_{i+1}}} N_{2i+1},W_{2i+1} 
\end{align*}
\end{definition}
The above defines a 
sequence of read and write operations
permitted by the policy starting from the given state.
The sequence starts by a subject reading from the initial object $o_1$
and reflects the possibility that any subsequent write operation by that subject is possibly influenced by information in $o_1$.
Any other subject that reads an object written to by $s_1$ may in turn be influenced by $o_1$ (even if the influence is inadvertent),
and hence any object they later write to 
may be influenced by confidential data in $o_1$.
Clearly, there can be many such flows starting in a given state.

The use of big-step transitions in the definition of a flow permits some operations that are not contributing directly to the given flow to occur in between operations that do contribute to the flow.
In this way, we range over arbitrary sequences of reads and writes containing a flow in Definition~\ref{def:if}.

\begin{figure}
\[
\begin{array}{rl}
\begin{array}{c}
  \begin{array}{c@{\;}c}    
    \vczero0 
    & \vcone1
    \\[1em]
  \vmod0{ ~\vozero0~ }
  &
  \vmod1{ ~\voone1\quad\votwo2\quad\vothree3~ }
  \end{array}
  \\[2em]
  \begin{array}{c@{\qquad\qquad}c}
  \vszero0
  &
  \vsone1
  \end{array}
\end{array}
  \Nedges{szero0/oone1,sone1/otwo2}
  \Wedges{szero0/otwo2}  
  \Redges{czero0/mod0,cone1/mod1}
\lts{\Read{\szero}{\ozero}} &
\begin{array}{c}
  \begin{array}{c@{\;}c}    
    \vczero0 
    & \vcone1
    \\[1em]
  \vmod0{ ~\vozero0~ }
  &
  \vmod1{ ~\voone1\quad\votwo2\quad\vothree3~ }
  \end{array}
  \\[2em]
  \begin{array}{c@{\qquad\qquad}c}
  \vszero0
  &
  \vsone1
  \end{array}
\end{array}
  \Nedges{szero0/otwo2,szero0/ozero0,szero0/oone1,sone1/otwo2}
  \Wedges{}  
  \Redges{czero0/mod0,cone1/mod1}
  \\
\lts{\Write{\sone}{\othree}}&
\begin{array}{c}
  \begin{array}{c@{\;}c}    
    \vczero0 
    & \vcone1
    \\[1em]
  \vmod0{ ~\vozero0~ }
  &
  \vmod1{ ~\voone1\quad\votwo2\quad\vothree3~ }
  \end{array}
  \\[2em]
  \begin{array}{c@{\qquad\qquad}c}
  \vszero0
  &
  \vsone1
  \end{array}
\end{array}
  \Nedges{szero0/otwo2,szero0/ozero0,szero0/oone1,sone1/otwo2}
  \Wedges{sone1/othree3}  
  \Redges{czero0/mod0,cone1/mod1}
\end{array}
\]
\caption{
Part of a flow between $\oone$ and $\othree$.
Prior steps can establish $\oone \leadsto \otwo$, via $\szero$.
Action $\Write{\sone}{\othree}$ completes the flow $\oone \leadsto \othree$, despite 
write access to $\otwo$ being revoked, since $\otwo$ may already be influenced by $\oone$.
}\label{fig:eg:transitive}
\end{figure}

For an example of a flow,
consider part of an 
information flow between $\oone$ and $\othree$
in Fig.~\ref{fig:eg:transitive}.
Beginning in state $\emptyset,\emptyset$,
the flow in question can be enabled by the following small step transitions. 
\[
\!
\begin{array}{rl}
\lts{\Read{\szero}{\oone}}
&
\left\{ (\szero,\oone) \right\}, \emptyset
\\
\lts{\Write{\szero}{\otwo}}
&
\left\{ (\szero,\oone), (\szero,\otwo) \right\},
\left\{ (\szero,\otwo) \right\}
\\
\lts{\Read{\sone}{\otwo}}
&
\left\{ (\szero,\oone), (\szero,\otwo), (\sone,\otwo) \right\},
\left\{ (\szero,\otwo) \right\} \text{\hspace{1.7cm}($\dagger$)}
\\
\lts{\Read{\szero}{\ozero}} &
\left\{ (\szero,\ozero), (\szero,\oone), (\sone,\oone), (\sone,\otwo) \right\},
\emptyset
\\
\lts{\Write{\sone}{\othree}} &
\left\{ (\szero,\ozero), (\szero,\oone), (\szero,\otwo), (\sone,\otwo), (\sone,\othree) \right\},
\left\{ (\sone,\othree) \right\}
\end{array}
\]
In the above, the state marked with ($\dagger$)
corresponds to the left hand side of Fig~\ref{fig:eg:transitive}.
The read operation following that state, also appearing in the figure,
is not active in the information flow under scrutiny,
since it follows a read.
It should be considered as part of a big step transition
comprising the final two operations together, where only the second is part of this particular flow.
That is, the final two operations above correspond to the big-step transition.
\[
\begin{array}{l}
\left\{ (\szero,\oone), (\szero,\otwo), (\sone,\otwo) \right\},
\left\{ (\szero,\otwo) \right\}
\\
\Lts{\Write{\sone}{\othree}} 
\left\{ (\szero,\ozero), (\szero,\oone), (\szero,\otwo), (\sone,\otwo),(\sone,\othree) \right\},
\left\{ (\sone,\othree) \right\}
\end{array}
\]
Notice that there is not a state of the system where all the read and write operations
between $\oone$ and $\othree$ are simultaneously enabled.
In contrast, the transitions in Fig.\ref{fig:eg:intransitive} 
are not part of a flow from $\ozero$ to $\otwo$.
This is because write access to $\otwo$ 
is revoked before $\ozero$ is read.
These two examples help explaining why the rich notion adopted from Kessler is appropriate for modelling information flow.

\begin{figure}
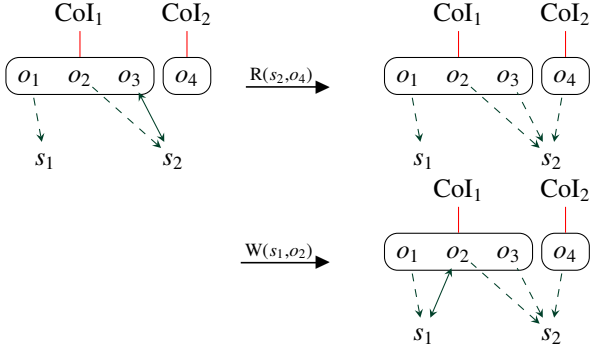

\[
\begin{array}{rl}
\begin{array}{c}
  \begin{array}{c@{\;}c}    
    \vczero0 
    & \vcone1
    \\[1em]
  \vmod0{ ~\vozero0\quad\voone1\quad\votwo2~ }
  &
  \vmod1{ ~\vothree3~ }
  \end{array}
  \\[3em]
  \begin{array}{c@{\qquad\qquad}c}
  \vszero0
  &
  \vsone1
  \end{array}
\end{array}
  \Nedges{szero0/ozero0,sone1/oone1}
  \Wedges{sone1/otwo2}  
  \Redges{czero0/mod0,cone1/mod1}
\lts{\Read{\sone}{\othree}}&
\begin{array}{c}
  \begin{array}{c@{\;}c}    
    \vczero0 
    & \vcone1
    \\[1em]
  \vmod0{ ~\vozero0\quad\voone1\quad\votwo2~ }
  &
  \vmod1{ ~\vothree3~ }
  \end{array}
  \\[3em]
  \begin{array}{c@{\qquad\qquad}c}
  \vszero0
  &
  \vsone1
  \end{array}
\end{array}
  \Nedges{sone1/otwo2,szero0/ozero0,sone1/oone1,sone1/othree3}
  \Redges{czero0/mod0,cone1/mod1}
  \\
\lts{\Write{\szero}{\oone}}&
\begin{array}{c}
  \begin{array}{c@{\;}c}    
    \vczero0 
    & \vcone1
    \\[1em]
  \vmod0{ ~\vozero0\quad\voone1\quad\votwo2~ }
  &
  \vmod1{ ~\vothree3~ }
  \end{array}
  \\[3em]
  \begin{array}{c@{\qquad\qquad}c}
  \vszero0
  &
  \vsone1
  \end{array}
\end{array}
  \Nedges{sone1/otwo2,szero0/ozero0,sone1/oone1,sone1/othree3}
  \Wedges{szero0/oone1}
  \Redges{czero0/mod0,cone1/mod1}
\end{array}
\]
\caption{
These transitions do not result in a flow from $\ozero$ to $\otwo$,
since write access to $\otwo$ is revoked before information about $\ozero$ flows via $\sone$.  
}\label{fig:eg:intransitive}
\end{figure}

\subsection{Mechanising confidentiality via \setlog and Coq}

Now that we have the missing ingredients to define Theorem~\ref{thm4},
we can proceed to mechanise the proof using \setlog and Coq.
The following recalls $\starProp$ from Sec.~\ref{sec:model}.
\begin{definition}\label{def:star}
For state $N, W$, we have
$*(N,W)$ whenever,
for all $(s,o) \in W$,
$s$, $o$ and $N$ satisfy the *-property (Def.~\ref{def:sp}).
\end{definition}
We already mentioned that the above is proven to be invariant in \setlog.
To be precise, the following expresses what is mechanised in \setlog.
\begin{lemma}
\label{lemma:setlog-star}
If $N, W \lts{\alpha} N', W'$
then $*(N, W) \Rightarrow *(N', W')$.
Furthermore, $N \subseteq N'$ and 
$\alpha = \Read{s}{o} \Rightarrow (s,o) \in N'$
and
$\alpha = \Write{s}{o} \Rightarrow (s,o) \in W'$.
Also, $*(\emptyset,\emptyset)$ holds.
\end{lemma}

We then observe that big-step transitions also preserve the properties 
that we guaranteed in Lemma~\ref{lemma:setlog-star}
\begin{corollary}\label{lemma:axioms}
If $N, W \Lts{\alpha} N', W'$
then $*(N, W) \Rightarrow *(N', W')$.
Furthermore, $N \subseteq N'$ and 
$\alpha = \Read{s}{o} \Rightarrow (s,o) \in N'$
and
$\alpha = \Write{s}{o} \Rightarrow (s,o) \in W'$.
\end{corollary}
\begin{proof}
Consider $N_0, W_0 \Lts{\alpha_n} N_{n+1}, W_{n+1}$,
and proceed by induction on the number of one-step transitions.
In the base case, there is a one-step transition, and hence the result follows immediately from Lemma~\ref{lemma:setlog-star}.
Consider the inductive case where 
$N_0,W_0 \Lts{\alpha_n} N_{n+1},W_{n+1}$
and $N_{n+1},W_{n+1} \lts{\alpha_{n+1}} N_{n+2},W_{n+2}$.
By the induction hypothesis, we have that $N_0 \subseteq N_{n+1}$
and $*(N_{0},W_{0}) \Rightarrow *(N_{n+1},W_{n+1})$.
Furthermore, by Lemma~\ref{lemma:setlog-star}, we have
$N_{n+1} \subseteq N_{n+2}$
and
$*(N_{n+1},W_{n+1}) \Rightarrow *(N_{n+2},W_{n+2})$,
and 
$\alpha_{n+1} = \Read{s}{o} \Rightarrow (s,o) \in N_{n+2}$
and
$\alpha = \Write{s}{o} \Rightarrow (s,o) \in W'$.
Thus 
$N_{0} \subseteq N_{n+2}$
and
$*(N_{0},W_{0}) \Rightarrow *(N_{n+2},W_{n+2})$, as required.
\end{proof}

Having introduced these preliminaries, we can prove Theorem~\ref{thm4}.
As explained in Sec.~\ref{sec:bridge},
the confidentiality 
property targeted by Brewer \& Nash 
essentially says that along any flow,
either the information flowing is sanitized
or it stays within the same dataset.
We use the Coq proof assistant 
to mechanise the
proof of the following intermediate theorem,
that proves that, in states where everything in $W$ satisfies the *-property,
the consequent of Theorem~\ref{thm4} holds.
Hence, since we have a mechanised proof that
the premise of Theorem~\ref{thm:Ricardo}
is an invariant in \setlog, then 
Theorem~\ref{thm4} follows immediately from 
Lemma~\ref{lemma:setlog-star} and Theorem~\ref{thm:Ricardo}.
\begin{theorem}\label{thm:Ricardo}
Consider any state $N_1, W_1$ such that $*(N_1,W_1)$.
If $o$ and $o'$ are objects,
then 
if $o \leadsto o'$ starts in $N_1, W_1$, then either $\ds{o} = \botDS$ or $\ds{o} = \ds{o'}$.
\end{theorem}
\begin{proof}
Assume that $*(N_1, W_1)$
and also assume that $o$ and $o'$ are objects
such that $o \leadsto o'$ starts in $N_1, W_1$.

By Definition~\ref{def:if}, since $o \leadsto o'$, we have some $n$ such that
$s_1,\dots,s_n$; $o_1,\dots,o_{n+1}$; $N_1,\dots,N_{2n+1}$; and $W_1,\dots,W_{2n+1}$ 
such that $N_{2i-1},W_{2i-1} \Lts{\Read{s_{i}}{o_{i}}} N_{2i},W_{2i} \Lts{\Write{s_i}{o_{i+1}}} N_{2i+1},W_{2i+1}$,
and $o = o_1$ and $o' = o_{n+1}$.

We then proceed by induction on $n$.
\begin{itemize}
\item \textsc{Base case, $n = 0$.} In this case $o \leadsto o$ and hence trivially $\ds{o} = \ds{o}$, as required.
%
%
\item \textsc{Induction hypothesis.} For for $n = k$ 
we have if $o \leadsto o_{k+1}$ then either $\ds{o} = \ds{o_{k+1}}$ or $\ds{o} = \botDS$.
\item \textsc{Inductive case, $n = k + 1$.} 

Notice that $o \flow o'$ can be decomposed into $o \flow o_{k+1} \flow o'$, where $o' = o_{k+2}$.
In turn, $o \flow o_{k+1}$ is of length $k$ so \by{induction hypothesis}
$\ds{o} = \ds{o_{k+1}} \vee \ds{o} = \botDS$.
If $\ds{o} = \botDS$ we are done immediately, hence we consider next
when
$\ds{o} = \ds{o_{k+1}}$.


We now aim to establish that either $\ds{o_{k+1}} = \Yo$ or $\ds{o_{k+1}} = \ds{o_{k+2}}$ holds.
Now,
 by Corollary~\ref{lemma:axioms}, we have
$*(N_{2i-1}, W_{2i-1}) \Rightarrow *(N_{2i}, W_{2i})$
and $*(N_{2i}, W_{2i}) \Rightarrow *(N_{2i+1}, W_{2i+1})$,
for $1 \leq i \leq k+1$.
Consequently, by transitivity repeatedly, we have
 $*(N_1, W_1) \Rightarrow *(N_{2k+3},W_{2k+3})$
and,
since we assumed that $*(N_1, W_1)$ holds,
we have that $*(N_{2k+3},W_{2k+3})$ holds.
Furthermore,
also by Corollary~\ref{lemma:axioms}, we have
$N_{2k+2} \subseteq N_{2k+3}$ and
 $(s_{k+1}, o_{k+1}) \in N_{2k+2}$ and $(s_{k+1}, o_{k+2}) \in W_{2k+3}$.
Hence, since
$N_{2k+2} \subseteq N_{2k+3}$
and 
$(s_{k+1}, o_{k+1}) \in N_{2k+2}$, we have 
$(s_{k+1}, o_{k+1}) \in N_{2k+3}$.

Now, since $*(N_{2k+3},W_{2k+3})$, by Def.~\ref{def:star},
since we have $(s_{k+1}, o_{k+2}) \in W_{2k+3}$,
it must be that $s_{k+1}$, $o_{k+2}$, and $N_{2k+3}$ satisfy the *-property.
Thus, by Def.~\ref{def:sp}, we have that 
 $\ds{o_{k+1}} = \ds{o_{k+2}}$ or $\ds{o_{k+1}} = \botDS$ holds,
 since we have just established that $(s_{k+1}, o_{k+1}) \in N_{2k+3}$.
 
Since we have just established that $\ds{o_{k+1}} = \Yo$ or $\ds{o_{k+1}} = \ds{o_{k+2}}$ holds
and we are considering the case when $\ds{o} = \ds{o_{k+1}}$,
we have that either $\ds{o} = \ds{o_{k+2}}$ or $\ds{o} = \botDS$ holds, as required.
\end{itemize}
\vspace{-1em}
\end{proof}

\subsection{On the mechanisation of the proof}

We summarise here how the above proofs
are mechanised in the 
replication package accompanying this
article~\cite{rep-package}.
The proofs, in Coq, of 
Corollary~\ref{lemma:axioms}
and Theorem~\ref{thm:Ricardo}
are aligned with the
proofs shown above.
Both proofs are established by induction,
the former over the length of a big-step transition,
and the latter over the number of big-step transitions comprising an information flow.
We do not automate fully these proofs in \setlog since the process for casting inductive proofs in \setlog currently would comprise manually casting the inductive steps as sub-problems without formal guarantees that the inductive conclusion follows from those sub-problems.
The reliance of the proof of Theorem~\ref{thm:Ricardo} on Corollary~\ref{lemma:axioms} is achieved by assuming appropriate axioms in Coq.
Similarly, the reliance of the proof of Corollary~\ref{lemma:axioms} on Lemma~\ref{lemma:setlog-star} is achieved by assuming axioms in Coq.
The proofs of Theorem~\ref{thm:Ricardo} and Corollary\ref{lemma:axioms} have not been attempted in \setlog
since their statements are formulas that do not fit in any of the
fragments of set theory for which \setlog implements decision procedures.
Before using \setlog to prove those formulas, we need to prove some decidability
results about a fragment of set theory resulting from the
combination of a few of its decidable fragments.

\section{Minimum subjects to access all datasets, mechanized}\label{sec:min}

In this section, we achieve full mechanisation of all theorems originally posed by Brewer \& Nash.
The missing proof, of Theorem~\ref{thm3}, ensures that whenever all datasets in a CoIC are accessed then,
there are at least as many subjects in the system as datasets.
This is clearly a special case of a stronger theorem stating that,
for any CoIC,
the number of datasets that have been accessed in that CoIC
is no greater than the total number of subjects in the system.
The proof of this theorem can be automatically handled by \setlog, by using a few tricks which we explain next. 

Let $\CoIC$ be the set of all possible CoICs and let $\mathcal{S}$ be the set of subjects in the system. 
Also, define the set of all datasets of CoIC $X$ accessed in state $N$ as follows.
\begin{equation}\label{th:dnx}
 D^{N}_X = \{Y \mid \exists (s,o) \in N: \coi{o} = X \land \ds{o} = Y\}
\end{equation}
The strengthening of the consequent of Theorem~\ref{thm3} mentioned above can be formulated as follows.
\begin{equation}\label{th:miss}
\forall X \in \CoIC: 
 \lvert D^N_X \rvert \leq \lvert\mathcal{S}\rvert
\end{equation}
To see why proving that the above is an invariant establishes Theorem~\ref{thm3},
consider $\XV$ as defined in Sec.~\ref{sec:bridge}.
Observe that 
if,
for all datasets $Y \in \XV$,
there exists subject $s$ and object $o$ such that $(s, o) \in N$ and $\ds{o} = Y$,
then we also have that $\XV = D^N_X$.

As we know from set theory, comparing the cardinality of sets, as in Eq~\eqref{th:miss}
can be achieved by exhibiting a surjective partial function
from $\mathcal{S}$ to $D^N_X$.
We can construct such a surjective partial function as follows:
\begin{equation}
f^{N}_X = \{(s,Y) \mid \exists (s,o) \in \N: \coi{o} = X \land \ds{o} = Y\}
\end{equation}
Since $f^N_X$ is surjective,
 a given dataset $Y$ can be accessed by more that one subject but a subject cannot access more than one dataset in a CoIC. 
Furthermore, since the range covers all datasets in the CoIC that have been accessed, there must be at least as many subjects as datasets being accessed in the CoIC.
Thus it remains to check only the following, to ensure our construction is correct ($\Ran$ denotes the co-domain of a relation).
\begin{equation}\label{th:miss2}
\forall X \in {} \CoIC: 
 \Pfun(f^{N}_X)
\land \Ran(f^N_X) = D^N_X
\end{equation}
We use several tricks to achieve the effect of verifying Eq~\eqref{th:miss2} in \setlog.
Firstly, a new state variable, $\RR$, is added to the model that was introduced in Sec.~\ref{sec:model}. $\RR$ is a set of ordered pairs of the form $(X,f)$ where $X \in \CoIC$ and $f$ is a set of ordered pairs $(s,Y)$ where $s$ is a subject and $Y$ a dataset. 
Second, the following invariant is added to the model to ensure that each $f$ is a partial function.
\begin{align*}
& \invariant(\siRR). \\
& \siRR(\RR) \text{ :- } \Foreach (X,F) \in \RR : \Pfun(F). 
\end{align*}
Thirdly, $\RR$ is updated whenever a read access is requested, to ensure that $f$ is kept in step with $f^N_X$ above.
 For instance, in $\strRead$ the update is performed by conjoining the following (recall that $S$ is the subject requesting access to an object whose label is $(X,Y)$):
\begin{align*}
& \Dom(\RR,A) \\
& \land (~~X \in A \\
& \qquad\land \ApplyTo(\RR,X,R) \\
& \qquad\land \Oplus(\RR,\{(X,\{(S,Y)~/~R\})\},\RR') \\
& \quad\lor X \notin A \\
& \qquad\land \RR' = \{(X,\{(S,Y)\})~/~\RR\})
\end{align*}
That is, $\RR$ is updated depending on whether $X$ is already in $\RR$'s domain or not. In the first case, $(S,Y)$ is added to the image of $X$ through $\RR$\footnote{$\Oplus$ is a \setlog constraint interpreted as B's \emph{overriding} operator.}, whereas in the second case $(X,\{(S,Y)\})$ is added to $\RR$. In other words, $\RR$ is updated in such a way that every time a subject $s$ reads from a new object $o$ whose security label is $(X,Y)$, then $(s,Y)$ is added to the image of $X$ through $\RR$. Put it in other way, in any state, if $X$ is a CoIC, then $\RR(X)$ is the set of pairs $(s,Y)$ such that subject $s$ is reading from some object $o$ such that $L(o) = (X,Y)$.

Finally, two more invariants are included ensuring that $\N$ and $\RR$ are always aligned. That is, the first invariant states that if $(X,f)$ belongs to $\RR$ and $(S,Y)$ belongs to $f$, then there exists $(S,(O,\ell))$ in $\N$ such that $\ell = (X,Y)$. The \setlog code is the following.
\begin{align*}
& \invariant(align\_\RR\_\N). \\
& align\_\RR\_\N(\RR,\N) \text{ :- }	\\
& \quad	\forall (X,F) \in \RR; (S_1,Y) \in F: \\ 
& \qquad \exists (S_2,(O,(X_1,Y_1))) \in \N:
      		S_2 = S_1 \land X_1 = X \land Y = Y_1.
\end{align*}
The second invariant (namely $align\_\N\_\RR$) states the opposite inclusion---i.e., if an ordered pair is in $\N$, then there's a corresponding ordered pair in $\RR$.



\paragraph*{Discussion on scope}
We have now fully interpreted and mechanised the original work of Brewer \& Nash. Table~\ref{table:1} provides hints to the reader about what to search for in the replication package~\cite{rep-package} to know the actual implementation of the mechanisation.
It is natural, in the future, to consider more comprehensive indicators that conflicts-of-interest are avoided.
For example,
we could check that there is never a flow to a subject from two objects in different datasets but the same CoIC.
This is stronger than Theorem~\ref{thm3}, since
we should prove that indirect flows from objects to subjects are also mitigated 
(Theorem~\ref{thm4} concerns flows from objects to objects).
The scope of the current paper however is complete, since we aimed 
to clarify and mechanise the original Brewer-Nash model.


\begin{table}[h]
\renewcommand{\arraystretch}{1.3}
\centering
\caption{Binding between theorem name and the counterparts used to achieve its automated verification}

\label{table:1}
\begin{tabular}{c|c|c}
\hline
\textbf{B\&N Theorem} & \textbf{Implemented by} & \textbf{Mechanised in} \\
\hline\hline
Theorem 1 &  t1 &  \\
 		  & theorem1 & \\
\cline{1-2}	  
Theorem 2 &  t2 & \\
 		  & theorem2 &  \\
\cline{1-2}
&  minSub & \\
Theorem 3 &  align\_Sds\_N & \setlog \\
&  align\_N\_Sds & \\
\cline{1-2}
& lemma1\_N\_spRead & \\
& lemma1\_N\_wkRead  & \\
& lemma1\_N\_rvkRead & \\
Theorem 4 & lemma1\_W\_write & \\
& lemma1\_N\_readWrite & \\
\cline{2-3}
&  Corollary\_1 & Coq \\
&  secure & \\

\hline
\end{tabular}
\vspace{10pt}

\end{table}

\section{Related and Future work: supporting policy makers}\label{sec:future}

There are refinements of Brewer-Nash and related security policy models that can be analysed, some already mentioned in the introduction.
This work can be seen as laying down a methodology that can be used to automate the analysis of such
security policy models.
Indeed, elements similar to our more explicit approach appear in Kessler~\cite{Kessler1992}, who treats access grants and operations separately.
A priority would be to adapt the model in this work to the conflict-of-interest relation of Lin~\cite{Lin1989}, mentioned in the introduction,
replacing the more restrictive Bell-LaPaudula-inspired CoIC labels of 
the Brewer-Nash that we have respected in this work.

In related work, the Bell-LaPadula policy model has already been verified in \setlog~\cite{DBLP:journals/jar/CristiaR21}. In that work, VCs were manually generated. The presence of the VCG in the current paper not only automates a nontrivial task but, mainly, increases confidence in the correctness of the VCs and the model, by reducing the possibility of human error in the toolchain.
There is also related work on using \setlog for verifying properties of the Android Permission System~\cite{Android} automating much of a 23KLOC Coq proof,
which is evidence that the methodology employed can scale.

As policy models can become complex we argue that the ``policy maker'' can benefit with efficient automated tools for analysing design decisions.
Consider for example Fig.~\ref{fig:eg:advanced},
where the right-hand side
can only be expressed using the more general conflict-of-interest relation due to Lin~\cite{Lin1989}.
Suppose a new policy model (not Brewer-Nash) permits a subject $\sone$ to write to $\oone$
while retaining read access to another dataset.
In this case, (1) a CoI relation itself is updated such that the dataset of $o_1$ absorbs the CoI of the dataset of $\otwo$,
and,
(2) since that would create violation of the simple security property, read access to $\oone$ is revoked entirely for $\stwo$.
That is, access for one subject 
is revoked due to actions of another subject and furthermore the whole system becomes more restrictive, leading to conflict resolution questions. 
The analysis of information flow properties
becomes trickier when read can be revoked since Lemma~\ref{lemma:setlog-star},
which assumes that read access monotonically increases, would be violated.

\begin{figure}[h!]
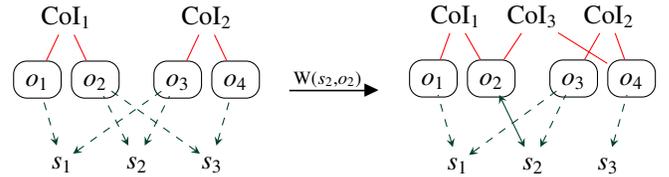

\[
\!\!\!\!\!\!
\begin{array}{c}
  \begin{array}{c@{\;}c}    
    \vczero0 & \vcone1
    \\[1em]
  \begin{array}{c@{~}c}
  \vmod0{ ~\vozero0~  } &
  \vmod1{  ~\voone1~ }
  \end{array}
  &
  \begin{array}{c@{~}c}
  \vmod2{ ~\votwo2~  }
  &
  \vmod3{  ~\vothree3~ }
  \end{array}
  \end{array}
  \\[3em]
  \begin{array}{c@{\qquad}c@{\qquad}c}
  \vszero0
  &
  \vsone1
  &
  \vstwo2 
  \end{array}
\end{array}
  \Nedges{szero0/ozero0,szero0/otwo2,sone1/oone1,sone1/otwo2,stwo2/othree3,stwo2/oone1}
  \Redges{czero0/mod0,czero0/mod1,cone1/mod2,cone1/mod3}
\!\!\!\!
\lts{\Write{\sone}{\oone}}
\!\!\!\!
\begin{array}{c}
  \begin{array}{c@{\quad}c@{\quad}c}    
    \vczero0 & 
    \vctwo2 &
    \vcone1
  \end{array}
    \\[1em]
  \begin{array}{c@{\;}c}    
  \begin{array}{c@{~}c}
  \vmod0{ ~\vozero0~  } &
  \vmod1{  ~\voone1~ }
  \end{array}
  &
  \begin{array}{c@{~}c}
  \vmod2{ ~\votwo2~  }
  &
  \vmod3{  ~\vothree3~ }
  \end{array}
  \end{array}
  \\[2em]
  \begin{array}{c@{\qquad}c@{\qquad}c}
  \vszero0
  &
  \vsone1
  &
  \vstwo2 
  \end{array}
\end{array}
  \Nedges{szero0/ozero0,szero0/otwo2,sone1/otwo2}
  \Nedges{stwo2/othree3}
  \Wedges{sone1/oone1}
  \Redges{czero0/mod0,czero0/mod1,cone1/mod2,cone1/mod3,ctwo2/mod1,ctwo2/mod3}
\]
\caption{
Can a policy allow $\sone$ to insist on write access to $\oone$?
Does this result in an updated conflict-of-interest relation
and also read access of $\stwo$ to $\oone$ being revoked?
}\label{fig:eg:advanced}
\end{figure}

In this work, we have extensively made use of \setlog to automatically verify
properties and theorems. The same tool allows policy makers to simulate the
behaviour of their policy model
before attempting any serious proofs without any extra effort.
In effect,
the tool can be used to retrieve the post-state after having executed one of the
specified operations (e.g. $\rRead$ or $\readWrite$) given a
particular pre-state. The replication package~\cite{rep-package}
includes guidelines to simulate some of the scenarios and examples included in
the paper. The same tool's feature may also be used to discover the
pre-state of an operation given the post-state (reverse simulation) or
verify if a property holds in a particular given state. 

Consider for example the pre-state shown in
Fig.~\ref{fig:eg:advanced}. If we want to verify whether $\starProp$ holds or
not, we can ask \setlog to solve the following.
\begin{align*}
& \N = \{(s_1,(o_1,(c_1,d_1),(s_1,(o_3,(c_2,d_3))),
	(s_2,(o_2,(c_1,d_2))),\\
&\qquad\quad\ (s_2,(o_3,(c_2,d_3))),
	(s_3,(o_2,(c_1,d_2))),(s_3,(o_4,(c_2,d_4)))\} \\
& \land \W = \emptyset \\
& \land \starProp(yo,\N,\W).
\end{align*}
Values of $N$ and $W$ are returned by \setlog,
  meaning $\starProp$ is
satisfiable, otherwise the answer would have been \textsf{no} (unsat).

There is related work on security policies using automated tools other than \setlog.
Some of those papers define a system model,
for example, in terms of temporal constraints in first-order logic~\cite{Brandt2010} or
using Z~\cite{Alam2016}.
The system model is then checked to determine whether it satisfies some formulation of the simple security rule.
In those two papers, the former formulates a variant of the *-property, while the latter appears to omit it entirely.
Our \setlog model could potentially also be combined with a system model
to check whether the system satisfies the simple security and *-property,
 similarly to such papers.
 However, the role of the current paper is complementary to that work,
since we are unaware of prior work mechanising properties of Brewer-Nash such as Theorems~\ref{thm1} to \ref{thm4}.
Variants of the *-property in the literature that we have alluded to 
could benefit from being justified by adapting our model and checking that appropriate variants of such theorems are discharged.


Future work includes verifying the relationship between the 
implicit model at the beginning of Sec.~\ref{sec:motivation},
and the explicit model in Fig.~\ref{fig:explicit-rules}.
Recall that the implicit model did not include the matrix $W$,
and $W$ was key for proving invariants establishing Lemma~\ref{lemma:setlog-star},
in addition to
making Brewer-Nash more implementable.

\begin{figure}[h!]
\begin{gather*}
\begin{prooftree}
\forall o' \colon (s, o') \in N \implies
\left(
  \ds{o'} = \ds{o}
  \vee
  \ds{o'} = \botDS
\right)
\justifies
N \lts{\Write{s}{o}} N 
\using iW
\end{prooftree}
\\[10pt]
\begin{prooftree}
\forall o' \colon (s, o') \in N \implies 
\left(
  \ds{o'} = \ds{o}
  \vee
  \ds{o'} = \botDS
\right)
\justifies
N \lts{\Write{s}{o}} N \cup \left\{ (s, o) \right\}
\using iRW
\end{prooftree}
\\[10pt]
\begin{prooftree}
\forall o' \colon (s, o') \in N \implies \ds{o'} = \ds{o} \vee \coi{o} \neq \coi{o'}
\justifies
N \lts{\Read{s}{o}} N \cup \left\{ (s, o) \right\}
\using iR
\end{prooftree}
\end{gather*}
\caption{Implicit rules for Brewer-Nash policies, including sanitized data.}\label{fig:implicit-rules}
\end{figure}

When enhanced with santized data, as shown in Fig.~\ref{fig:implicit-rules},
the implicit and explicit models align in the sense that any operation in one is possible in the other.
They are even bisimilar, an observation guaranteeing that results concerning information flow are preserved in the implicit model.
We leave these formal comparisons as future work.

\section{Conclusion}

The theme of this agenda is to equip policy makers such that they may make bolder well-informed decisions regarding policies,
 preserving confidentiality constraints on information while potentially increasing access.
How can such policy makers be sure that their policies preserve their intended information flow properties?
More specifically, in this work, we argue that the widespread usage of Chinese Wall policies 
 and high-stake consequences of policy failure,
 mean that we should not rely solely on the definitions and original proofs of Brewer \& Nash
 and we should bring them up to the level of assurance given by modern mechanised tools.
 Indeed, we have mechanised in \setlog invariants formulated in terms of the simple security rule and *-property (Theorem~\ref{thm1}, Lemma~\ref{lemma:setlog-star} \& Sec.~\ref{sec:model});
  and also that the number of subjects in a system cannot be less than the number of datasets accessed in each conflict of interest class (Theorem~\ref{thm3} \& Sec.~\ref{sec:min}).
The steps of Theorem~\ref{thm4} mechanised in Coq are expressed in Theorem~\ref{thm:Ricardo}.
  
 We have deliberately stuck closely to the original Brewer-Nash security policy model in this work.
 This is to remove any doubt that we verify anything other than the core model proposed by Brewer \& Nash,
 and
 also because we believe that, even for that model, the operational semantics of access was left somewhat open to interpretation, as elaborated on in Sections~\ref{sec:motivation} and~\ref{sec:bridge}.
 Indeed, in Section~\ref{sec:motivation} we have pointed out that interpretations are not unique
 (e.g., if write access can be granted without read access, that opens up an initial phase where data can be pushed from subjects to multiple datasets independently of a conflict of interest, before read is granted within a confidential dataset, resulting in write being revoked elsewhere).
Furthermore, the space of existing and future extensions of Brewer-Nash is large, as touched on in Section~\ref{sec:future},
and the systematic exploration of that space is open to creativity, where debates may be further substantiated by adapting the methodology employed in the current paper.


\bibliographystyle{IEEEtran}
\bibliography{references.bib}

\end{document}